\documentclass[11pt]{article}

 \usepackage[margin=1in]{geometry}

\usepackage[utf8]{inputenc}
\usepackage[english]{babel}
\usepackage[colorlinks]{hyperref}

\usepackage{amsmath}
\usepackage{amsfonts}
\usepackage{amssymb}
\usepackage{amsthm}
\usepackage{mathtools}

\usepackage{graphicx}
\usepackage{extarrows}
\usepackage[linesnumbered,ruled]{algorithm2e}
\usepackage{xcolor}
\usepackage{multicol}
\usepackage{environ}
\usepackage[T1]{fontenc}
\usepackage{csquotes}
\usepackage{orcidlink}
\usepackage{authblk}
\usepackage{tikz}
\usetikzlibrary{backgrounds,calc,decorations.pathreplacing, hobby, arrows,arrows.meta,decorations.pathreplacing,decorations.pathmorphing,shapes.symbols}
\usetikzlibrary{shapes}
\usepackage{framed}
\usepackage[most]{tcolorbox}
\usepackage{xspace}
\usepackage{mathtools}
\usepackage{booktabs}
\usepackage{tabularx}
\usepackage{enumitem}
\usepackage{thm-restate}

\hypersetup{
	linkcolor=rwthred,
	citecolor=rwthgreen,
	urlcolor=rwthblue,
	menucolor=rwthred
}

\usepackage[noabbrev,capitalise]{cleveref}

\let\emptyset\varnothing

\newtheorem{theorem}{Theorem}[section]
\newtheorem{lemma}[theorem]{Lemma}

\newcommand{\w}[2]{\ensuremath{\textcolor{leadercolor}{#1}, \textcolor{followercolor}{#2}}}

\newcommand{\tspan}[4][black]{\draw[decorate,decoration={brace,amplitude=7pt,raise=4pt}] (#2, -1.5) -- (#3, -1.5) node[midway,yshift=20pt,align=center] {#4};}

\newcommand{\ew}[2]{\ensuremath{(\textcolor{leadercolor}{#1}, \textcolor{followercolor}{#2})}}

\newcommand{\interval}[6][black]{\draw[Bracket-Bracket,thick,#1] (#3, -#2*.7-1.3) -- node[above=-2pt,black]{#5} node[below=-1pt]{#6} (#4, -#2*.7-1.3);}

\newenvironment{ichart}[3][1]{
	\begin{tikzpicture}[#1]
		\def\yb{-#2}
		\draw[-latex'] ($(-.5, \yb-.7)$) -- ($(#3, \yb-.7) + (.5, 0)$);	
		\foreach \t in {0, ..., #3}{
			\draw[-] ($(\t, \yb-.7)$) -- ($(\t, \yb-.7) + (0, -.1)$);        
		}
	}{\end{tikzpicture}}

\tikzset{
	small vertex/.style={draw,circle,thick,fill=black,minimum size=1.2mm,inner sep=0pt},
	node/.style={draw, circle, fill, minimum size=.1cm,inner sep=0pt},
	directed/.style={->, draw, thick},
	sink/.style={rwthturquoise},
	source/.style={rwthorange},
	bluevertex/.style={circle, fill= rwthblue, inner sep=0pt,minimum size=6pt},
	redvertex/.style={circle, fill= rwthred, inner sep=0pt,minimum size=6pt},
	greenvertex/.style={circle, fill= rwthgreen, inner sep=0pt,minimum size=6pt},
	leader_node/.style={circle, draw=black, fill=leader_color, minimum size=6mm, inner sep=0pt},
	follower_node/.style={circle, draw=black, fill=follower_color, minimum size=6mm, inner sep=0pt},
	tiny vertex/.style={draw,black,circle,inner sep=0pt,minimum size=3pt,fill},
	box/.style={draw,gray,fill=none,rounded corners},
	curly/.style={-,decorate,decoration={snake,amplitude=.6mm,segment length=2mm,post length=0mm}}	
}

\newcommand{\NN}{{\mathbb{N}}}
\newcommand{\ZZ}{{\mathbb{Z}}}

\newcommand{\RR}{{\mathbb{R}}}
\renewcommand{\phi}{\varphi}
\newcommand{\np}{\ensuremath{\mathsf{NP}}\xspace}

\newcommand{\p}{\ensuremath{\mathsf{P}}\xspace}

\newcommand{\forward}{(\( \Rightarrow \))\xspace}
\newcommand{\backward}{(\( \Leftarrow \))\xspace}
\newcommand{\commentt}[1]{}
\newcommand\set[1]{\{#1\}}
\newcommand{\abs}[1]{\left|#1\right|}

\DeclareMathOperator*{\argmax}{\arg\max}
\DeclareMathOperator*{\argmin}{\arg\min}

\newcommand{\sigmap}{\ensuremath{\Sigma_2^\mathsf{p}}\xspace}

\newcommand{\independentset}{{\sc Independent Set}\xspace}
\newcommand{\IS}{{\sc Interval Selection}\xspace}
\newcommand{\intervalselectionshort}{\textsc{ISel}\xspace}

\newcommand{\V}{\ensuremath{V_\ell \ \dot\cup \ V_f}}

\newcommand{\vset}{\ensuremath{X \ \dot\cup \ Y}}
\renewcommand{\i}{i}
\newcommand{\I}{I}
\newcommand{\Iset}{\ensuremath{I_\ell \ \dot\cup \ I_f}}

\newcommand{\bisch}{{\sc Bilevel Interval Selection}\xspace}
\newcommand{\BISch}{{\sc BISel}\xspace}
\newcommand{\BIS}{{\sc Bilevel Independent Set}\xspace}
\newcommand{\bis}{{\sc BIS}\xspace}

\newcommand{\bcnf}{\ensuremath{B_2^{\text{CNF}}}\xspace}
\newcommand{\optimistic}{\ensuremath{\mathsf{o}}\xspace}
\newcommand{\pessimistic}{\ensuremath{\mathsf{p}}\xspace}
\newcommand{\costleader}{\ensuremath{\mathrm{c}}}
\newcommand{\costfollower}{\ensuremath{\mathrm{d}}}
\DeclareMathOperator{\opt}{opt}

\DeclareMathOperator{\sol}{Sol^{ISel}}

\newtheoremstyle{named}{}{}{\itshape}{}{\bfseries}{.}{.5em}{\thmnote{#3's }#1 (#2)}
\theoremstyle{named}

\makeatletter
\newcommand{\problemtitle}[1]{\gdef\@problemtitle{#1}}
\newcommand{\probleminput}[1]{\gdef\@probleminput{#1}}
\newcommand{\problemquestion}[1]{\gdef\@problemquestion{#1}}
\NewEnviron{problemenv}{
	\problemtitle{}\probleminput{}\problemquestion{}
	\BODY
	\par\addvspace{.5\baselineskip}
	\noindent
	\begin{tabularx}{\textwidth}{|@{\hspace{5pt}} l X c}
		\multicolumn{2}{@{}l}{\sc\@problemtitle} \\
		\textbf{Given:} & \@probleminput \\
		\textbf{Task:} & \@problemquestion
	\end{tabularx}
	\par\addvspace{.5\baselineskip}
}
\makeatother
\makeatletter \g@addto@macro\@floatboxreset\centering \makeatother

\usepackage{xcolor}
\definecolor{rwthblue}{RGB}{0,84,159}
\definecolor{rwthlightblue}{RGB}{142,186,229}
\definecolor{rwthred}{RGB}{204,7,30}
\definecolor{rwthgreen}{RGB}{87,171,39}
\definecolor{rwthorange}{RGB}{246,168,0}
\definecolor{rwthmagenta}{RGB}{227,0,102}
\definecolor{rwthyellow}{RGB}{255,237,0}
\definecolor{rwthpetrol}{RGB}{0,97,101}
\definecolor{rwthturquoise}{RGB}{0,152,161}
\definecolor{rwthmay}{RGB}{189,205,0}
\definecolor{rwthbordeaux}{RGB}{161,16,53}
\definecolor{rwthviolet}{RGB}{97,33,88}
\definecolor{rwthpurple}{RGB}{122,111,172}

\colorlet{leader_color}{rwthblue!65} 
\colorlet{follower_color}{rwthred!65} 

\colorlet{leadercolor}{rwthblue}
\colorlet{followercolor}{rwthred}

\title{On the Complexity of Bilevel Independent Set Problem}

\author[1]{Komal Muluk~\orcidlink{0009-0002-3915-7198}}

\affil[1]{Department of Mathematics, TU Dortmund University, Germany
\texttt{komal.muluk@math.tu-dortmund.de}}

\date{\vspace{-2\baselineskip}}

\begin{document}
	
	\maketitle
	
	\begin{abstract}
		We consider a bilevel optimization problem in which the ground set is partitioned between two decision makers, a leader and a follower, whose optimization problems are interleaved. 
		We study the \textsc{Bilevel Independent Set} problem, and its special case, the \textsc{Bilevel Interval Selection} problem, on different variants emerging from a combination of the type of leader's objective function, the type of follower's objective function, and the setting in which the follower reacts, i.e., either optimistically or pessimistically. Here we consider sum and bottleneck type objective functions. We investigate the computational complexity of all these variants for the \textsc{Bilevel Independent Set} problem, and sort them into their respective level of the polynomial hierarchy.	Our results range from $\mathsf{P}$, $\mathsf{NP}$-completeness to $\Sigma_2^\mathsf{p}$-completeness. For the \textsc{Bilevel Interval Selection} problem, we give a dynamic programming algorithm running in time $\mathcal{O}(n^4\log n)$ for the variants in which the leader and the follower have objective functions of the sum type.
	
		\medskip\noindent
		\emph{Keywords:} bilevel optimization, dynamic programming, \np-hardness, \sigmap-hardness, computational complexity
	\end{abstract}

	\section{Introduction}
\label{sec:introduction}

In the recent years, the field of bilevel optimization has garnered numerous attention~\cite{BroMarSav2008, dempe2015bilevel, dempe2020bilevel, woeginger2021trouble, beck2023, grune2025completeness}.
A typical problem in bilevel optimization consist of a system of two optimization problems which bear a hierarchical structure.
One of the problems is nested into the other, that is, one optimization problem resides in the constraint set of the other.
The goal of a bilevel optimization problem is that both decision makers, despite having potentially distinct objective functions, work together to build a feasible solution of the underlying problem, all the while trying to optimize their own objective function.
Commonly, the two decision makers who decide on optimizing the two optimization problems are called \emph{leader} and \emph{follower}, and they are considered to be non-cooperative.
Both decision makers have their own set of decision variables and their own private weight function over all decision variables.
When the set of decision variables controlled by the leader and the follower is completely disjoint, the respective problem is known as a \emph{partitioned-items bilevel optimization problem}~\cite{henke25}.
In this study, we will investigate the partitioned-items bilevel independent set problem, and its special case, the bilevel interval selection problem.
In the \textsc{Independent Set} problem, given a graph $G$ and a weight function $w\colon V(G)\to \RR_+$, the task is to find a maximum weighted set $S\subseteq V(G)$ such that no two vertices in $S$ are adjacent in $G$.
Such a set $S$ is called an \emph{independent set} of the graph $G$.

Even though, bilevel optimization is currently a rapidly growing field, the traces of this concept dates back to 1930 in the duopoly model of von Stackelberg~\cite{Stackelberg1934}.
Jeroslow in 1985, obtained the first complexity results for bilevel optimization problems; his results imply that bilevel linear programs with continuous variables are \np-hard and that bilevel integer linear programs are even \sigmap-hard~\cite{jeroslow1985polynomial}.

The formulation of our bilevel optimization problem is as follows.
Consider a combinatorial optimization problem with the set of feasible solutions $\mathcal{F}\subseteq 2^E$ over a ground set~$E$.
Due to the partitioned-items formulation, $E$ is partitioned into two subsets, $E=E_\ell\ \dot\cup\ E_f$ where~$E_\ell$ is controlled by the leader and $E_f$ is controlled by the follower. 
Additionally, we have weight functions $w_\ell, w_f \colon E\to \RR_+$, $w_\ell$ for the leader and $w_f$ for the follower.
Moreover, we have two objective functions~$\costleader$ and~$\costfollower$ for the leader and the follower, respectively. 
Where, typically, the leader's objective function, $\costleader$, depends on her weight function $w_\ell$, for example $\costleader(L\cup F) = \sum_{v\in L\cup F} w_\ell(v)$, 
whereas the follower's objective function, $\costfollower$, depends on his weight function $w_f$.
Then our bilevel optimization problem is formulated as follows:
\begin{subequations}\label{eq:bo:bilevelproblem}
	\begin{alignat}{4}
		&\max_{L\subseteq E_\ell}\quad &&\costleader(L \cup F)
		\label{eq:bo:1a}\\
		&&&\text{where $F\subseteq E_f$ solves the follower's problem}
		\nonumber\\[1.5ex]
		&&&{\max_{F\subseteq E_f}}^*\quad\costfollower(L\cup F)
		\label{eq:bo:1b}\\
		&&&\phantom{{\max_{F\subseteq E_f}}^*\quad} \text{such that $L\cup F \in \mathcal{F}$}\nonumber
	\end{alignat}
\end{subequations}

The optimization problem mentioned in \eqref{eq:bo:1a} is the \emph{leader's optimization problem} and the one mentioned in \eqref{eq:bo:1b} is the \emph{follower's optimization problem}.
In general, the goal is to solve the leader's optimization problem subject to the condition that the follower chooses his set of variables, i.e., $F\subseteq E_f$, that optimizes his own objective function.
The leader has to choose a set $L\subseteq E_\ell$ which in the end can be augmented by a suitable follower's reaction $F$ such that $L\cup F$ belongs to the family $\mathcal{F}$ of feasible solutions.
An action $L \subseteq E_\ell$ of the leader is called a \emph{leader's feasible action} if there exists a follower's reaction $F$ such that $L\cup F\in \mathcal{F}$.
Given a leader's feasible action $L\subseteq E_\ell$, we denote the set of the follower's suitable reactions to it by $\Omega(L) \coloneqq \{F\subseteq E_f \mid L\cup F\in \mathcal{F}\}$.
Now, given the set~$L$, a follower's reaction $F'\subseteq E_f$ is called a \emph{follower's possible reaction} if $F' \in \argmin_{F\in \Omega(L)} \costfollower(L\cup F)$.
Observe that for a specific leader's action $L$, the follower may have multiple possible reactions. 
In this case, we consider two settings in which the follower reacts.
First, we consider the \emph{optimistic} setting in which the follower always chooses, among his possible reactions the one which benefits the leader the most. 
Second, we consider the \emph{pessimistic} setting, in which he reacts by choosing the set~$F$ which is the worst for the leader.
The asterisk $(*)$ in the follower's problem in \eqref{eq:bo:1b} encodes the follower's behavior, either optimistic $(\optimistic)$ or pessimistic $(\pessimistic)$.
Given a leader's action $L$, we denote the \emph{follower's optimal reaction} to it (according to \optimistic or \pessimistic setting) by~$F_L$.
Note that, in case the follower's optimal reaction is also not unique, all such reactions evaluate to the same leader's objective value and the same follower's objective value. 
So, from the optimization perspective, a specific optimal reaction of the follower does not make any difference, thus without loss of generality, we say \emph{the} follower's optimal reaction.

Furthermore, we consider two types of objective functions, up to optimistic or pessimistic setting, namely the sum type and the bottleneck type. 
For a maximization problem, the leader's sum type function ($\costleader_s$) and bottleneck type function ($\costleader_b$) can be defined as follows:
\begin{equation*}
	\costleader_s(S) = \sum_{v \in S} \quad w_\ell(v) \qquad\text{and}\qquad
	\costleader_b(S) = \min_{v \in S} \quad w_\ell(v),
\end{equation*}
where $S$ is a subset of the ground set~$E$.
We define the follower's objective functions $\costfollower_s(S)$ and $\costfollower_b(S)$ analogously; 
there we use the follower's weight function $w_f$ instead of $w_\ell$.
\begin{equation*}
	\costfollower_s(S) = \sum_{v \in S} \quad w_f(v) \qquad\text{and}\qquad
	\costfollower_b(S) = \min_{v \in S} \quad w_f(v).
\end{equation*}
Throughout this paper, we discuss eight different variants of bilevel problems that emerge from choosing: 
\begin{itemize}
	\item the leader's objective function from $\costleader_s$ or $\costleader_b$,
	\item the follower's objective function from $\costfollower_s$ or $\costfollower_b$, and
	\item the behavior of the follower from $\optimistic$ or $\pessimistic$.
\end{itemize}
For simplicity, we represent each of these variants using a $3$-tuple, each component of the tuple capturing a binary choice from the above three factors.
Thus, we represent a variant as $(\costleader,\costfollower,*)$, wherein $\costleader\in\{\costleader_s,\costleader_b\}$, $\costfollower\in\{\costfollower_s,\costfollower_b\}$, and $*\in \{\optimistic, \pessimistic\}$.

A real-life scenario which involves a partitioned-items bilevel problem can be narrated as follows:
Think of a factory with only one machine and two operators who have completely disjoint skill sets.
There are several jobs (with designated time slots) that can be performed at this machine, but only one job can be run at a time.
Each job has a tight time window in which it can be performed.
Due to the mutually exclusive skill sets of the operators, each job can be performed by either of the operators, but not both.
Each operator benefits differently from a job that gets done (they also benefit from each other's jobs). 
Their individual goal is to optimize their personal total payoff.
In order to schedule the jobs on the machine, a specific, say more senior operator, chooses a set of her jobs that she would like to execute during their respective time slots, so she reserves the machine for those times.
Then the second operator arranges a suitable subset of his jobs that he wants to perform, while making sure that the time windows of his jobs do not clash with the jobs of the other operator.
This is one possible bilevel variant of the \IS problem.
We refer to this problem as simply \bisch (\BISch, as short) and discuss it thoroughly in \Cref{sec:bo:bilevel-interval-scheduling}.
In order to view an instance of this problem, we refer the reader to \Cref{fig:bo:BISche:instance-illustration}.

\subsection{Associated Work}
\label{subsec:bo:realted-work}

In the past few years, numerous articles have been devoted to studying various partitioned-items bilevel optimization problems.	
Starting from the simplest case of the \textsc{Bilevel Selection} problem~\cite{henke2024robust}, the study has expanded to the bilevel versions of more complex problems.
For the underlying \textsc{Selection} problem, given a cost function over a finite set of items and a positive integer $k$, the task is to select exactly $k$ items which minimize the total cost.
In~\cite{henke2024robust}, Henke also deals with the \emph{robust} bilevel version of the problem.

In 2009, Gassner and Klinz studied the partitioned-items \textsc{Bilevel Assignment} problem on the eight variants that we also discuss in this paper~\cite{GasKli2009}. 
They showed \np-hardness for all but one of the variants. 
Later in 2022, Fischer, Muluk, and Woeginger resolved the complexity of last open variant and showed it to be \np-hard as well~\cite{fischer2022note}.

Moreover, the \textsc{Bilevel Knapsack} problem, a generalizations of the \textsc{Bilevel Selection} problem, has been shown to be $\sigmap$-hard for three different variants, one of which is a partitioned-items problem formulation~\cite{caprara2014study}.
The authors of this study also derived some approximation results for the three bilevel knapsack variants studied; they showed two of their variants do not allow polynomial-time approximation algorithms with a worst-case guarantee (unless $\p=\np$), however, for the third variant, they were able to derive a polynomial-time approximation scheme~\cite{caprara2013complexity}. 
Carvalho, Lodi, and Marcotte gave a polynomial-time algorithm to solve a continuous bilevel knapsack problem~\cite{carvalho2018polynomial}, which was further improved by Fischer and Woeginger~\cite{fischer2020faster}.
Yet another study on robust bilevel continuous knapsack problem was done by Buchheim and Henke~\cite{buchheim2022robust}.

Moreover, the partitioned-items version of the \textsc{Bilevel Spanning Tree} problem has been studied in~\cite{buchheim2022complexity};
the authors also considered the eight variants of the problem as we do, and they showed that some variants are \np-hard, while some are in~\p.
Additionally, a few other bilevel versions of the spanning tree problem have been studied in~\cite{shi2019bilevel, shi2023mixed}.
A recent addition to the list is a study on the partitioned-items \textsc{Bilevel Shortest Path} problem, wherein Henke and Wulf completely classified the problem variants into classes of the polynomial hierarchy~\cite{henke25}.

\subsection{Our Contribution}
\label{subsec:bo:contribution}

In this paper, we provide a variety of computational complexity results ranging over polynomial-time algorithms, \np-hardness, and even \sigmap-hardness, for the \bisch (\BISch as short) and \BIS (\bis as short) problem in \Cref{sec:bo:bilevel-interval-scheduling} and \Cref{sec:bo:BIS}, respectively.

An instance of the \IS (\intervalselectionshort) 
problem consists of a set of intervals $I$ on the real line and a weight function $w\colon I\to \RR_+$.
The task is to find a subset $I'\subseteq I$ of pairwise non-intersecting intervals that maximizes the total sum of weights within $I'$.
Moreover, it is well-known that \textsc{Interval Selection} can be modeled solving \textsc{Independent Set} on interval graphs. 
Thus, we first consider the \bisch problem in \Cref{sec:bo:bilevel-interval-scheduling}.
We show that \BISch is tractable when the leader's objective function and the follower's objective function are of sum type, in both the optimistic and the pessimistic setting.
We give a dynamic programming algorithm that runs in $\mathcal{O}(n^4\log n)$ time, where $n$ is the total number of intervals in the given instance.

For the \BIS problem, however, we start with a general study on simple undirected graphs.
We get the results mentioned in \Cref{results:overview} for all eight variants that emerge from considering $(\costleader_s/\costleader_b, \costfollower_s/\costfollower_b, \optimistic/\pessimistic)$ on general graphs, as in simple undirected graphs, and on bipartite graphs. 
\begin{table}
	\centering
	\begin{tabular}{l|c|c|c|c}
		\toprule
		Graph class & Leader & Follower & Setting & Result \\
		\midrule
		Interval graphs & $\costleader_s$ & $\costfollower_s$ & $\optimistic/\pessimistic$ & \p(\Cref{thm:bo:BISche:algo-correct})\\
		General & $\costleader_s/\costleader_b$ & $\costfollower_s$ & $\optimistic /\pessimistic$ & $\sigmap$-complete (\Cref{thm:general:sigmap2}) \\
		General & $\costleader_b$ & $\costfollower_b$ & $\pessimistic$ & \np-complete (\Cref{thm:c_bd_bp}) \\
		General & $\costleader_s$ & $\costfollower_b$ & $\optimistic/\pessimistic$ & \np-complete (\Cref{thm:bis:gen:np-hard:c_sd_bp/o}) \\
		General/Bipartite & $\costleader_b$ & $\costfollower_b$ & $\optimistic$ & \p (\Cref{thm:general:P}) \\
		Bipartite/Planar & $\costleader_s$ & $\costfollower_s$ & $\optimistic/\pessimistic$ & \np-complete (\Cref{thm:bipartite:c_sd_so/p})\\
		Bipartite & $\costleader_s$ & $\costfollower_b$ & $\optimistic$ & \p (\Cref{thm:bipartite:c_sd_bo}) \\
		Bipartite & $\costleader_s$ & $\costfollower_b$ & $\pessimistic$ &  \p (\Cref{thm:bipartite:c_sd_bp}) \\
		Bipartite & $\costleader_b$ & $\costfollower_s$ & $\optimistic/\pessimistic$ & \np-complete (\Cref{thm:bipartite:c_bd_bp:c_bd_so/p}) \\
		Bipartite & $\costleader_b$ & $\costfollower_b$ & $\pessimistic$ & \np-complete (\Cref{thm:bipartite:c_bd_bp:c_bd_so/p}) \\
		\bottomrule
	\end{tabular}	
	\caption{Computational complexity of \bis for various cases formed according to the nature of the leader and follower's objective functions and the follower's behavior}
	\label{results:overview}
\end{table}

\section{Preliminaries}
\label{sec:preliminaries}

For a graph $G$, we denote the \emph{vertex set} of the graph by $V(G)$ and the \emph{edge set} by $E(G)$. 
We sometimes omit $G$ and simply write $V$ and $E$ instead.
An edge incident on vertices $u$ and $v$ is represented by $\{u,v\}$.
We represent the \emph{open neighborhood} of a vertex $v\in V(G)$ by $N_G(v)\coloneqq \{u\in V(G) \mid \{u,v\}\in E(G)\}$, whereas the \emph{closed neighborhood} of $v$ by $N_G[v] \coloneqq N(v) \cup \{v\}$.
We denote the \emph{bipartition} of a set $S$ into two sets $A$ and $B$ by $S=A\ \dot\cup\ B$.
The \emph{cardinality} of a set $S$ is denoted by $|S|$.
Further, we denote the set of \emph{natural} numbers (including zero) by $\NN$, the set of \emph{integers} by $\ZZ$, and the set of \emph{real} numbers by $\RR$.
Moreover, $\RR_+$ represent the subsets of non-negative elements of $\RR$.
For $n\in \NN$, define $[n] \coloneqq \{1, 2, \ldots, n\}$.
Given two constants $C_1, C_2\in \RR_+$, we denote $C_1$ is \emph{much greater than} (or \emph{much less than}) $C_2$ with notation $C_1 \gg C_2$ (or $C_1\ll C_2$, respectively).

We use the notation $\leftarrow$ to define an update operation on a set (sometimes also on a graph), for example, in order to update a set $A$ with the addition of an element $b$, we may write, $A\leftarrow A\cup \{b\}$.

	\section{Bilevel Interval Selection}\label{sec:bo:bilevel-interval-scheduling}
In this paper, we study the variants $(\costleader_s, \costfollower_s, \optimistic/\pessimistic)$ of the \BISch problem according to the Problem \eqref{eq:bo:bilevelproblem}, and we refer to its single-level version as the \emph{underlying \IS} problem.
Frank, in 1976, gave a dynamic program for solving \intervalselectionshort optimally~\cite{frank}.
This is useful when we need to use the solution of the \intervalselectionshort problem as a subroutine in our result. 
Let us now define the \bisch(\BISch) problem.

Let $\I = \set{\i_1, \i_2, \ldots, \i_n}$ be a set of intervals in which each interval $\i_j = [a_j, b_j)$ consists of a start point $a_j$ and an end point $b_j$ such that $a_j, b_j \in \RR_+$, and $a_j < b_j$. 
Let $\pi \coloneqq (\i_1, \i_2, \ldots, \i_n)$ be the sequence of intervals in $\I$ according to the non-decreasing ordering of their end points.
So $b_j \leq b_{j'}$ for all $j \leq j'$.
Define $\I_k \coloneqq \{\i_1, \i_2, \ldots, \i_k\}$ for any $k\in[n]$.
The intervals in $I$ are partitioned into two sets, $\I = \Iset$, where $\I_\ell$ is the set of intervals controlled by the leader and $\I_f$ is controlled by the follower.
Both players have their own weight functions $w_\ell, w_f \colon \I \to \RR_+$.
The objective functions of both the leader and the follower are of sum type and they are $\costleader_s(L \cup F) = \sum_{\i\in L \cup F} w_\ell(\i)$ and $\costfollower_s(L \cup F) = \sum_{\i\in L \cup F} w_f(\i)$, respectively.
The ultimate task is for the leader to choose her subset $L \subseteq \I_\ell$ of pairwise-disjoint intervals, that the follower later augments with his chosen subset $F \subseteq \I_f$ of additional pairwise-disjoint intervals---all the while trying to optimize his own objective function $\costfollower_s(L \cup F )$---such that, in the end, the leader achieves the best possible weight for her objective function, i.e., $\costleader_s(L \cup F)$ is maximized. 
Here, we consider the variants $(\costleader_s, \costfollower_s, \optimistic/\pessimistic)$ of the \bisch(\BISch) problem.

\begin{subequations}\label{prob:bo:is}
	\begin{alignat}{4}
		&\max_{L\subseteq \I_\ell}\quad && \sum_{\i\in L \cup F} w_\ell(\i)
		\label{prob:bo:is:1a}\\
		&&&\text{where $F\subseteq \I_f$ solves the follower's problem}
		\nonumber\\[1.5ex]
		&&&{\max_{F\subseteq \I_f}}^*\quad \sum_{\i\in L \cup F} w_f(\i)
		\label{prob:bo:is:1b}\\
		&&&\phantom{{\max_{F\subseteq \I_f}}^*\quad} \text{such that $L\cup F$ is a set of pairwise-disjoint intervals in $\I$,}\nonumber
	\end{alignat}
\end{subequations}
where $*\in \{\optimistic,\pessimistic\}$, and thus encodes the behavior or the follower, either optimistic or pessimistic.
Refer to \Cref{fig:bo:BISche:instance-illustration} for an illustration of an instance of the \BISch problem.
\begin{figure}[t]
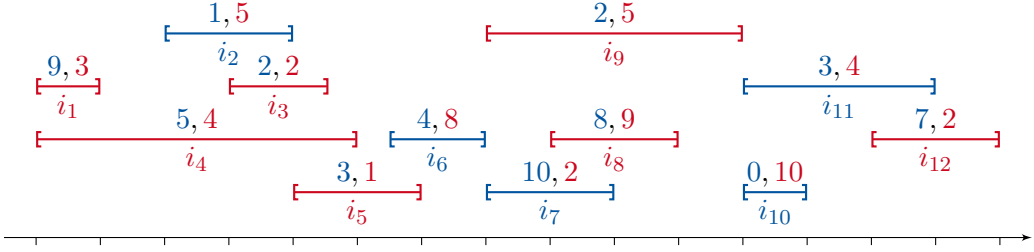

	\begin{center}
		\begin{ichart}[xscale=0.85]{4}{15}
			\interval[followercolor]{2}{0}{1}{\w{9}{3}}{$i_1$}
			\interval[leadercolor]{1}{2}{4}{\w{1}{5}}{$i_2$}
			\interval[followercolor]{2}{3}{4.55}{\w{2}{2}}{$i_3$}
			\interval[followercolor]{3}{0}{5}{\w{5}{4}}{$i_4$}
			\interval[followercolor]{4}{4}{6}{\w{3}{1}}{$i_5$}
			\interval[leadercolor]{3}{5.5}{7}{\w{4}{8}}{$i_6$}
			\interval[leadercolor]{4}{7}{9}{\w{10}{2}}{$i_7$}
			\interval[followercolor]{3}{8}{10}{\w{8}{9}}{$i_8$}
			\interval[followercolor]{1}{7}{11}{\w{2}{5}}{$i_9$}
			\interval[leadercolor]{4}{11}{12}{\w{0}{10}}{$i_{10}$}
			\interval[leadercolor]{2}{11}{14}{\w{3}{4}}{$i_{11}$}
			\interval[followercolor]{3}{13}{15}{\w{7}{2}}{$i_{12}$}
		\end{ichart}
	\end{center}
	\caption{Example of an instance of the \BISch problem. 
		Intervals controlled by the leader and the follower are drawn in blue and red colors, respectively.
		The corresponding leader's and follower's weights for each interval are also shown in blue and red colors, respectively.
	}
	\label{fig:bo:BISche:instance-illustration}
\end{figure}

Define $p(k)$ to be the largest possible index smaller than $k$ such that, for the pair $i_k$ and $i_{p(k)}$ of intervals, we have $b_{p(k)} \leq a_k$. 
That means $i_{p(k)}$ is the last interval in the ordering which completely ends before the interval $i_k$ starts.
If no such interval exists, then set $p(k)\coloneqq 0$.

\subsection{A Dynamic Program}
\label{subsec:bo:IS:dp}

We devise a dynamic programming algorithm to solve the \BISch problem for the $(\costleader_s, \costfollower_s, \optimistic/\pessimistic)$ variants.
Later in \Cref{lem:bo:BISche:o/p-follower}, we prove that one can incorporate the follower's behavior, either optimistic or pessimistic, in the follower's weight function; In \Cref{lem:bo:BISche:o/p-follower}, we provide strategies to update $w_f$ for this purpose.
Thus throughout this section, we assume that, the follower always chooses according to his behavior, either optimistic or pessimistic.

\subsubsection{Description of the Algorithm}\label{alg:dp:BilevelIntervalScheduling}

We begin by giving a recursive formulation to calculate the optimal objective value of the leader in an instance restricted to the intervals $\I_k = \{\i_1, \i_2, \ldots, \i_k\}$. 
Let us define $\opt[k]$
as the maximum value that the leader can achieve on an instance of interval set~$\I_k$.
Assume that we have precalculated all $\opt[k']$ for all $k' < k$, then we show that we can calculate the value $\opt[k]$ in polynomial time using the precalculated values.
Let $i_0\coloneqq [-\delta,0)$ for a very small value $\delta > 0$ with weights $w_\ell(i_0)\coloneqq 0$ and $w_f(i_0)\coloneqq 0$.
We define our base case as:	$\opt[0] \coloneqq 0$.

Further, we claim that $\opt[k]$ satisfies the following recurrence relation.
\[
\opt[k] = \begin{cases}
	\max \Big \{ w_\ell(\i_k) +\opt[p(k)],\ \opt[k-1] \Big \} & \quad \text{if } \i_k \in \I_\ell,\\
	\max_{\substack{j < k \\ \i_j \in \I_\ell\cup\{i_0\}}} \left\{\opt[p(j)] + w_\ell(\i_j) +  \sum_{\i \in \sol(\I_f^{j,k}, w_f)} w_\ell(\i) \right\} & \quad \text{if } \i_k \in \I_f,\\
\end{cases}
\]
where, $\I_f^{j,k}\coloneqq \set{\i \in \I_f \mid \i \in \I_k \setminus \I_{j} \text{ and } \i \cap \i_j = \emptyset}$, and the set $\sol(\I_f^{j,k}, w_f)$ is a maximum weighted subset of pairwise-disjoint intervals in $\I_f^{j,k}$ according to the follower's weight function~$w_f$, which can be computed using Frank's dynamic program~\cite{frank}.
We now prove the correctness of our recurrence formulation for $\opt[k]$.

\subsubsection{Analysis of the Algorithm}
\begin{lemma}\label{lem:bo:BISche:reccurence-correctness}
	The function $\opt[k]$ describes the leader's optimal objective value when only the subset of intervals $\set{\i_1, \i_2, \ldots, \i_k}$ is considered.
\end{lemma}
\begin{proof}
	When $k=0$, there are no intervals to choose from. 
	Thus each set of pairwise-disjoint intervals is an empty set and the leader's objective function evaluates to $0$. 
	As we defined $\opt[0]\coloneqq 0$, this case holds.
	Otherwise $k\neq 0$, in that case either $i_k\in \I_\ell$ or $i_k\in \I_f$.
	
	First, consider the case when $i_k\in \I_\ell$.
	Let $I^\star\subseteq \{i_1, i_2, \dots, i_k\}$ be a set of pairwise-disjoint intervals which optimizes the leader's objective function in the \BISch problem.
	Thus $I^\star$ consists of the leader's action $L$ along with the follower's optimal reaction $F_L$ to $L$.	
	For such an $I^\star\subseteq \{i_1, i_2, \dots, i_k\}$, 
	either $i_k\in I^\star$ or $i_k\notin I^\star$.
	(i) If $i_k\in I^\star$, then the rest of $I^\star$ satisfies $I^\star\setminus \{i_k\} \subseteq \{i_1, i_2, \dots, \i_{p(k)}\}$ as $i_j\cap i_k \neq \emptyset$ for all $p(k) < j \leq k$.
	Recall that $\{i_1, i_2, \dots, \i_{p(k)}\} \eqqcolon \I_{p(k)}$.
	The set $I^\star\setminus \{i_k\}$ maximizes the leader's objective function over the set of intervals $\I_{p(k)}$, because if not, then there is yet another set $I'\subseteq\I_{p(k)}$ which maximizes the leader's objective function over $\I_{p(k)}$. 
	Then we can construct a set of pairwise-disjoint intervals $\I'\cup \{i_k\}$, which returns a leader's objective value greater than that of $I^\star$.
	Since the set $\I'\cup \{i_k\}$ is also a feasible solution to the \BISch problem over the interval set~$I_k$, this contradicts the optimality of $\I^\star$.
	Thus, we conclude that $I^\star\setminus \{i_k\}$ maximizes the leader's objective function over the set of intervals $\I_{p(k)}$
	which is captured in $\opt[p(k)]$.
	So, $\sum_{i\in \I^\star} w_\ell(i) = \opt[p(k)]+w_\ell(i_k)$.
	(ii) If $i_k\notin I^\star$, then $I^\star\subseteq \I_{k-1}$.
	We claim that $\I^\star$ also maximizes the leader's objective function over $\I_{k-1}$.
	If not, then a subset $I''\subseteq I_{k-1}$, $\I^\star\neq I''$ that optimizes the leader's objective function in $I_{k-1}$ can replace $I^\star$ to return a better leader's valuation also for set~$I_k$, which again contradicts the optimality of $I^\star$.
	Thus, we conclude that $\I^\star$ maximizes the leader's objective function over $\I_{k-1}$.
	Consequently, the leader's value obtained by $I^\star$ can be captured by $\opt[k-1]$.
	Thus in this case, $\sum_{i\in \I^\star} w_\ell(i) = \opt[k-1]$.
	As a result, if $i_k\in I_\ell$, then 
	\[
	\opt[k] = \max \Big \{ w_\ell(\i_k) +\opt[p(k)],\ \opt[k-1] \Big \}.
	\]		
	Now, consider the remaining case when $i_k\in \I_f$.
	Once more, let $I^\star\subseteq \{i_1, i_2, \dots, i_k\}$ be a set of pairwise-disjoint intervals which optimizes the leader's objective function.
	Since $i_k\notin I_\ell$, consider the largest index $j<k$, such that $i_j\in I^\star\cap I_\ell$.
	If no such $j$ exists, then none of the leader's intervals are a part of the optimal solution $I^\star$, and we assume $j=0$ with $w_\ell(i_0)\coloneqq 0$, $w_f(i_0)\coloneqq 0$, $\opt[0] \coloneqq 0$, and $p(0)\coloneqq 0$.
	Since, $I^\star \ni i_j$, we claim that $I^\star \subseteq I_{p(j)} \cup \{i_j\} \cup I^{j,k}_f$, where $I_{p(j)} \coloneqq \{i_1, i_2, \dots, i_{p(j)}\}$ and $\I_f^{j,k}\coloneqq \set{\i \in \I_f \mid \i \in \I_k \setminus \I_{j} \text{ and } \i \cap \i_j = \emptyset}$ (also, define $\I_f^{0,k}\coloneqq \I_f\cap \I_k$).
	
	The condition $I^\star \subseteq I_{p(j)} \cup \{i_j\} \cup I^{j,k}_f$ holds because, if $j$ is the largest index such that $i_j\in \I^\star\cap I_\ell$, then no other interval that intersects $i_j$ belongs to $\I^\star$.
	Moreover, each $i_{j'} \in I^\star$ with $j'>j$ also belongs to the set~$I_f$.
	Note that the three sets $I_{p(j)}, \{i_j\}$, and $I^{j,k}_f$ are pairwise-disjoint, in the sense that they do not share common intervals. 
	For an illustration, see \Cref{fig:bo:BISche:partitioning-of-intervals}. 
	\begin{figure}[t]
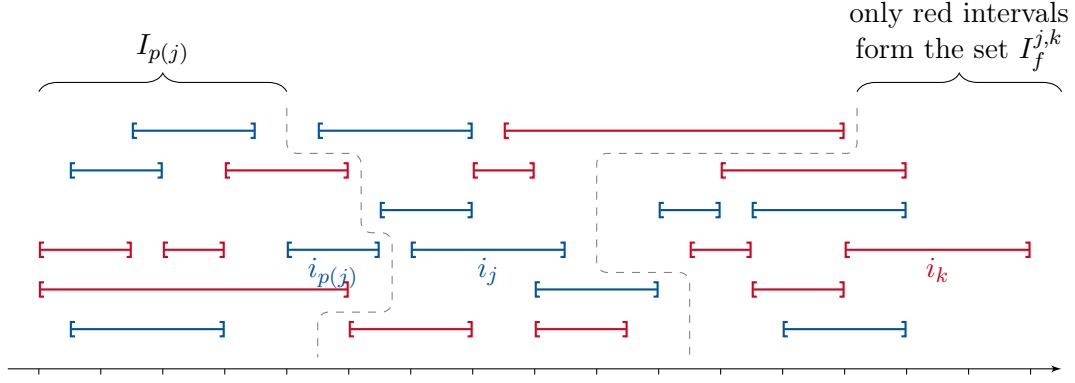

		\begin{center}
			\begin{ichart}[xscale=0.82,yscale=.75]{5.5}{16}
				\interval[followercolor]{5}{0}{5}{}{}
				\interval[followercolor]{4}{0}{1.5}{}{}
				\interval[followercolor]{4}{2}{3}{}{}
				\interval[followercolor]{2}{3}{5}{}{}
				
				\interval[followercolor]{6}{5}{7}{}{}
				\interval[followercolor]{6}{8}{9.5}{}{}
				\interval[followercolor]{2}{7}{8}{}{}
				\interval[followercolor]{1}{7.5}{13}{}{}
				
				\interval[followercolor]{2}{11}{14}{}{}
				\interval[followercolor]{4}{10.5}{11.5}{}{}
				\interval[followercolor]{4}{13}{16}{}{$i_k$}
				\interval[followercolor]{5}{11.5}{13}{}{}
				
				\interval[leadercolor]{6}{.5}{3}{}{}
				\interval[leadercolor]{4}{4}{5.5}{}{$i_{p(j)}$}
				\interval[leadercolor]{2}{.5}{2}{}{}
				\interval[leadercolor]{1}{1.5}{3.5}{}{}
				
				\interval[leadercolor]{1}{4.5}{7}{}{}
				\interval[leadercolor]{3}{5.5}{7}{}{}
				\interval[leadercolor]{4}{6}{8.5}{}{$i_j$}
				\interval[leadercolor]{5}{8}{10}{}{}
				
				\interval[leadercolor]{3}{10}{11}{}{}
				\interval[leadercolor]{3}{11.5}{14}{}{}
				\interval[leadercolor]{6}{12}{14}{}{}
				
				\tspan{0}{4}{$I_{p(j)}$}
				\draw[dashed,gray,rounded corners] (4, -1.6) -- (4, -2.4) -- (5.2, -2.4) -- (5.2, -3.8) -- (5.7, -3.8) -- (5.7, -5.2) -- (4.5, -5.2) -- (4.5, -6);
				\draw[dashed,gray,rounded corners] (13.2, -1.6) -- (13.2, -2.4) -- (9, -2.4) -- (9, -4.5) -- (10.5, -4.5) -- (10.5, -6);
				\tspan{13.2}{16.5}{only red intervals\\form the set $I_f^{j,k}$\\[.1em]\ }
			\end{ichart}
		\end{center}
		\caption{When $i_k\in I_f$: partitioning of the intervals across an interval $i_j\in I_\ell$. 
			When $i_j$ belongs to the leader's action $L\subseteq I_\ell$, the sets $I_{p(j)}$ and $I^{j,k}_f$ are completely disjoint. The intervals from $I_{p(j)}$ do not intersect with $i_j$, nor do they intersect with any interval in $I^{j,k}_f$. Similarly, $i_j$ does not intersect with any interval from $I^{j,k}_f$.}
		\label{fig:bo:BISche:partitioning-of-intervals}
	\end{figure}
	Additionally, each pair of intervals $i$ and $i'$ belonging to two distinct sets from $I_{p(j)}, \{i_j\}$, and $I^{j,k}_f$ is non-intersecting.
	Now, we claim that the following hold:
	\begin{enumerate}[label=(\alph*)]
		\item $\sum_{i\in I^\star\cap I_{p(j)}} w_\ell(i) = \opt[p(j)]$,\label{claim:a}
		\item $\sum_{i\in I^\star\cap \{i_j\}} w_\ell(i) = w_\ell(i_j)$, and\label{claim:b}
		\item $\sum_{i\in I^\star\cap I^{j,k}_f} w_\ell(i) = \sum_{i\in \sol(I^{j,k}_f, w_f)} w_\ell(i)$,\label{claim:c}
	\end{enumerate}
	where $\sol(\I_f^{j,k}, w_f)$ is the solution of \intervalselectionshort problem on $\I_f^{j,k}$ according to function $w_f$.
	
	Let us start with claim \ref{claim:a}.
	Note that the set $I^\star\cap I_{p(j)}$ can be built by the leader-follower interplay of the \BISch problem.
	Assume that claim \ref{claim:a} does not hold.
	Then either $\sum_{i\in I^\star\cap I_{p(j)}} w_\ell(i) < \opt[p(j)]$ or $\sum_{i\in I^\star\cap I_{p(j)}} w_\ell(i) > \opt[p(j)]$.
	For the former case, consider $(I^\star\setminus I_{p(j)}) \cup I_{\opt[p(j)]}$, where $I_{\opt[p(j)]}$ is the subset of pairwise-disjoint intervals in $I_{p(j)}$ that corresponds to the optimal value $\opt[p(j)]$.
	The set $(I^\star\setminus I_{p(j)}) \cup I_{\opt[p(j)]}$ is also a feasible solution for the \BISch problem on the interval set~$I_k$. 
	Moreover, 
	\[
	\sum_{i\in (I^\star\setminus I_{p(j)}) \cup I_{\opt[p(j)]}} w_\ell(i)~ > ~\sum_{i\in I^\star} w_\ell(i),
	\]
	as $\sum_{i\in I^\star\cap I_{p(j)}} w_\ell(i) < \opt[p(j)]$.
	This contradicts the optimality of $I^\star$.
	
	As for the latter case, if $\sum_{i\in I^\star\cap I_{p(j)}} w_\ell(i) > \opt[p(j)]$, 
	this contradicts with the definition of $\opt[k]$ because  $\opt[p(j)]$ is the maximum value that the leader can achieve on the interval set~$I_{p(j)}$.
	Thus, we conclude that \ref{claim:a} always holds.
	
	For claim \ref{claim:b}, note that the chosen index $j<k$ is the last index that satisfies $i_j\in I^\star\cap I_\ell$. 
	Indeed, since $i_j\in I^\star$, \ref{claim:b} holds trivially.
	
	Now, we discuss claim \ref{claim:c}.
	Since the leader's feasible action $(I^\star\cap I_\ell) \subseteq I^\star\setminus I^{j,k}_f$, and $i_j\in I^\star\cap I_\ell$, none of the intervals in $I^{j,k}_f$ overlaps with an interval from the leader's action $I^\star\cap I_\ell$.
	Additionally, as we noticed earlier, the sets $I^{j,k}_f$ and $I_{p(j)}$ are disjoint, and the range of their intervals is non-overlapping.
	This allows the follower to solve his problem independently on each of these sets.		
	Altogether, the follower's optimal reaction to $I^\star \cap I_\ell$ includes solving for the max-weighted subset of intervals, according to the weight function $w_f$, in $I^{j,k}_f$ and $I_{p(j)}$ which do not overlap with intervals in $I^\star \cap I_\ell$.
	To obtain such a subset of intervals over~$I^{j,k}_f$, the follower simply solves the underlying \IS problem on an instance with interval set~$I^{j,k}_f$ and weight function $w_f$.
	Thus, the leader's objective function in $I^\star \cap I^{j,k}_f$ must evaluate to the value $\sum_{i\in \sol(I^{j,k}_f, w_f)} w_\ell(i)$.
	Thus, claim \ref{claim:c} holds.
	
	As a result of \ref{claim:a}, \ref{claim:b}, and \ref{claim:c}, when $i_k\in I_f$, we enumerate over all $j<k$ such that $i_j\in I_\ell$ can be the last interval of the leader's action in set~$I_k$.
	Then the optimal solution of \BISch on interval set~$I_k$ lies within one of these cases.
	As a result, when $i_k\in I_f$, we calculate the value of $\opt[k]$ by using the following recursive formula:
	\[
		\opt[k] = \max_{\substack{j < k \\ \i_j \in \I_\ell}} \left\{\opt[p(j)] + w_\ell(\i_j) +  	\sum_{\i \in \sol(\I_f^{j,k}, w_f)} w_\ell(\i) \right\}. \label{eq:interval} \qedhere
	\] 
\end{proof}

\begin{restatable}[]{lemma}{ISchOptimisticPessimistic}
	\label{lem:bo:BISche:o/p-follower}
	Given a leader's action $L\subseteq I_\ell$, the follower's optimal reaction $F_L\subseteq I_f$ can be calculated in polynomial time both for the optimistic and the pessimistic behavior of the follower. 
\end{restatable}
\begin{proof}
	In order to capture the essence of the optimistic or pessimistic behavior of the follower in \eqref{prob:bo:is}, we update the follower's weight function $w_f$ according to the following rules. For all $i\in I$, we do:
	\begin{itemize}
		\item if the follower behaves optimistically, update,
		\[
		w_f(i) \leftarrow w_f(i) + \epsilon\cdot w_\ell(i),
		\]
		\item if the follower behaves pessimistically, update,
		\[
		w_f(i) \leftarrow w_f(i) - \epsilon\cdot w_\ell(i),
		\]
	\end{itemize}
	where $\epsilon>0$ is a very small number which slightly modifies the follower's weight function to capture the behavior of the follower.
	
	Note that, for any subset $I'\subseteq I$, 
	\[
	\sum_{i\in I'} w_f(i)\leftarrow \sum_{i\in I'} \big(w_f(i) \pm \epsilon\cdot w_\ell(i) \big).
	\]
	Since $\epsilon$ is a very small number, 
	\[
	\argmax_{I'\subseteq I} ~ \sum_{i\in I'} \big(w_f(i) \pm \epsilon\cdot w_\ell(i) \big) \subseteq \argmax_{I'\subseteq I} ~ \sum_{i\in I'} w_f(i).
	\]
	So, the follower indeed always chooses a set $F\subseteq I_f$ which optimizes his objective function.
	Moreover, due to the additional $\pm \sum_{i\in I'} \epsilon\cdot w_\ell(i)$ term---whenever the follower is faced with multiple possible reactions to $L$---he additionally maximizes $ \sum_{i\in I'} \epsilon\cdot w_\ell(i)$ in the optimistic setting and minimizes $\sum_{i\in I'} \epsilon\cdot w_\ell(i)$ in the pessimistic setting.
	That means, when the follower has multiple possible reactions, in the optimistic setting, he acts in favor of the leader and in the pessimistic setting he acts against the interest of the leader.
\end{proof}

\begin{theorem}\label{thm:bo:BISche:algo-correct}
	The algorithm described in \Cref{alg:dp:BilevelIntervalScheduling} returns the optimal value for the \BISch problem given in \eqref{prob:bo:is} and it runs in polynomial time in the size of the input instance.
\end{theorem}
\begin{proof}
	We gave a recurrence formula to calculate the optimal objective value of the \BISch problem when the instance consists of intervals only within the set $I_k = \{i_1, i_2, \dots, i_k\}$.
	We store this optimal value in $\opt[k]$.
	As shown in \Cref{lem:bo:BISche:reccurence-correctness}, our recurrence relation is correct, and it indeed maintains the leader's optimal objective value within interval set~$I_k$.
	Furthermore, it is easy to verify that we implement this exact recurrence in the algorithm.
	We use the bottom-up approach in the execution of our algorithm.
	That means, $\opt[j]$ for all $j<k$ is calculated prior to calculating the value of $\opt[k]$.
	All in all, the output of the algorithm returns $\opt[n]$, which is the optimal value of the \BISch problem on the instance with interval set $\{i_1,i_2, \dots, i_n\}$.
	
	Now, for the running time of the algorithm: It takes $\mathcal{O}(n\log n)$ time to arrange the intervals in non-decreasing order of their end points, another $\mathcal{O}(n\log n)$ to calculate the indices $p(j)$ for all $j\in [n]$. 
	For each pair $i<j$, it takes $\mathcal{O} (n^2\log n)$ to calculate the set~$I^{i,j}_f$ and the corresponding value of $\sol(I^{i,j}_f, w_f)$. 
	Moreover, there are ${n\choose 2}$ many $(i,j)$ pairs.
	Thus it takes altogether $\mathcal{O} (n^4\log n)$ time.
	Furthermore, it takes $\mathcal{O}(n^2)$ to execute the recursive formula in the bottom-up manner.
	Hence, the algorithm has time complexity $\mathcal{O} (n^4\log n)$.
\end{proof}

So far, we studied the bilevel versions of the \IS problem, which is also the \independentset problem on interval graphs.
Thus, we now generalize this study and consider the \BIS (\bis) problem. 

	\section{Bilevel Independent Set}\label{sec:bo:BIS}

Unlike the polynomial-time results for \BISch, here we show a variety of hardness results for different variants of \BIS.
When considering sum type objective functions for both leader and follower, the problem difficulty increases significantly from being polynomial-time solvable for the \BISch to being \sigmap-hard for \bis.

An instance of \BIS (\bis) consists of a graph $G = (\V, E)$, where $V_\ell$ and $V_f$ are the sets of vertices controlled by the leader and the follower, respectively. Moreover, let $w_\ell, w_f \colon \V \to \RR_+$ be their weight functions, respectively.
Then the \bis problem is to optimize the following objective function.
\begin{subequations}\label{prob:bo:bis}
	\begin{alignat}{4}
		&\max_{L\subseteq V_\ell}\quad &&\costleader(L \cup F)
		\label{eq:bis:1a}\\
		&&&\text{where $F\subseteq V_f$ solves the follower's problem}
		\nonumber\\[1.5ex]
		&&&{\max_{F\subseteq V_f}}^*\quad\costfollower(L\cup F)
		\label{eq:bis:1b}\\
		&&&\phantom{{\max_{F\subseteq V_f}}^*\quad} \text{such that $L\cup F$ is an independent set and $L\cup F\neq \emptyset$.}\nonumber
	\end{alignat}
\end{subequations}
where $(^*)$ encodes the behavior of the follower, either optimistic or pessimistic.
We again note that we are dealing with a maximization problem.
Additionally, note that, we impose an extra condition that $L\cup F\neq \emptyset$ in our problem. 
This is because, if we allow $L\cup F=\emptyset$, then in case of bottleneck functions for both decision makers, since the minimum over the empty set is $\infty$, 
both the leader and the follower will always choose the null set.

We study the above formulation given in \eqref{prob:bo:bis} of the \bis problem on the class of simple undirected graphs and bipartite graphs.
In the upcoming section, we begin the study of \bis on the simple undirected graph class.

\subsection{Complexity on Simple Undirected Graphs}
\label{subsec:bo:bis:general-graphs}

In this section, we study the $8$ different variants of the \bis problem and settle their computational complexity on simple undirected graphs.
For a overview of our results, refer to \Cref{results:overview}.

To begin with, we study the variants of the Problem \eqref{prob:bo:bis}, where the leader's objective function is of sum or bottleneck type, so $\costleader \in \{\costleader_s,\costleader_b\}$, the follower's objective function is of sum type, so $(\costfollower_s)$, and the follower acts as per the optimistic or pessimistic setting, so $\optimistic$ or $\pessimistic$.	
Throughout this subsection, we work towards proving the result stated in \Cref{thm:general:sigmap2}.	
For all four variants $(\costleader_s/ \costleader_b, \costfollower_s, \optimistic/\pessimistic)$, we give a unified reduction from the following $\sigmap$-complete problem~\cite{johannes2011}:

\begin{problemenv}
	\problemtitle{\textbf{Problem:} \bcnf}
	\probleminput{A boolean formula $\phi$ in CNF that contains exactly three literals per clause and consists of a variable set which is partitioned into \vset.}
	\problemquestion{Does there exist a truth assignment to the variables in $X$ such that there is no truth assignment to the variables in $Y$ that would help satisfy the formula $\phi$?}
\end{problemenv}

\textbf{Construction of an instance of \bis:}
Let us consider an instance $I$ of \bcnf which consists of a boolean formula $\phi$ with two sets of variables $X = \set{x_1, x_2, \ldots, x_{n_1}}$ and $Y = \set{y_1, y_2, \ldots, y_{n_2}}$, and the clause set $C = \set{C_1, \ldots, C_m}$, where each clause $C_i$ contains exactly three literals. 
Now, we construct an instance $I' = (\V, E, w_\ell, w_f)$ of \bis as follows (also see \Cref{fig:reduction:general}):
\begin{itemize}
	\item For each variable $x_i \in X$, the set~$V_\ell$ contains vertices $a_i, \bar{a}_i$. Denote the set of all these vertices by $A$.
	\item For each variable $y_i \in Y$, the set~$V_f$ contains vertices $b_i, \bar{b}_i$. Denote the set of all these vertices by $B$.
	\item For each clause $C_i = (l_i^1 \vee l_i^2 \vee l_i^3) \in C$, add vertices $c_i^1, c_i^2, c_i^3, c_i$ to set~$V_f$.
	\item For each $i \in [n_1]$, add an edge $\{a_i,\bar{a}_i\}$ to $E$.
	For each $i \in [n_2]$, add an edge $\{b_i,\bar{b}_i\}$ to $E$.
	\item For each $i \in [m]$, add edges $\{c_i^1,c_i^2\}, \{c_i^1,c_i^3\}, \{c_i^2,c_i^3\}, \{c_i^1,c_i\}, \{c_i^2,c_i\}, \{c_i^3,c_i\}$ to $E$.
	\item For all $i \in [m]$, if $l_i^1$ is a positive literal, say $l_i^1 = x_j$, then add edge $\{c_i^1, \bar{a}_j\}$ to $E$;
	add edge $\{c_i^1, \bar{b}_j\}$ to $E$ when $l_i^1 = y_j$.
	Similarly, if $l_i^1$ is a negative literal, say $l_i^1 = \bar{x}_j$, then add edge $\{c_i^1, {a_j}\}$ to $E$.
	Also add $\{c_i^1, b_j\}$ to $E$ when $l_i^1 = \bar{y}_j$.
	Add similar edges corresponding to all three literals $l_i^1, l_i^2, l_i^3$ of the clause $C_i$.
	\item Finally, add a vertex $z$ to set~$V_f$ and make it adjacent to $c_i$ for all $i\in [m]$.
\end{itemize}

\begin{figure}[t]
	\begin{center}
		\begin{tikzpicture}[scale=0.62, >=latex]
			\def\leadervertexCount{3}
			\def\followervertexCount{3}
			\def\clauseCount{4}
			\tikzset{
				leader_node/.style={circle, draw=black, fill=leader_color, minimum size=5mm, inner sep=0pt},
				follower_node/.style={circle, draw=black, fill=follower_color, minimum size=5mm, inner sep=0pt}
			}
			
			\node (phi) at (-6,7) {$\begin{aligned}
					\phi &= (x_3 \vee x_2 \vee \bar{x}_1)\\
					&\quad\wedge (\bar{x}_2 \vee  \bar{x}_3 \vee \bar{y}_1)\\
					&\quad\wedge (x_1 \vee x_3 \vee y_3)\\
					&\quad\wedge (y_1 \vee y_2 \vee \bar{y}_3)
				\end{aligned}$};
			
			\node (to) at (-2.5, 7) {$\longrightarrow$};
			\begin{scope}[shift={(-1, 12)},rotate=-90]
				\foreach \i in {1,...,\leadervertexCount} {
					\node[leader_node] (L\i*2) at (1,7+\i*2) {\small $a_{\pgfmathprint{int(\leadervertexCount+1-\i)}}$}; 
					\node[leader_node] (L\i*2-1) at (1,7+\i*2-1) {\small $\bar{a}_{\pgfmathprint{int(\leadervertexCount+1-\i)}}$}; 
					\draw (L\i*2) -- (L\i*2-1);
				}
								
				\foreach \i in {1,...,\followervertexCount}	{
					\node[follower_node] (F\i*2) at (1,7-\i*2) {\small $\bar{b}_\i$};
					\node[follower_node] (F\i*2+1) at (1,7-\i*2+1) {\small $b_\i$};
					\draw (F\i*2) -- (F\i*2+1);
				}
								
				\draw[black!20, dashed, rounded corners] (0.4,7.5) -- (0.4,13.5) --(1.6,13.5)-- (1.6,7.5) -- cycle;
				\draw[black!20, dashed, rounded corners] (0.4,6.5) -- (0.4,0.5) --(1.6,0.5)-- (1.6,6.5) -- cycle;
				\draw[black!20, dashed, rounded corners] (5.2,13.5) -- (6.8,13.5) -- (6.8,0) -- (5.2,0)-- cycle;
				\draw[black!20, dashed, rounded corners] (7.5,12.5) -- (8.5,12.5) -- (8.5,1) -- (7.5,1)-- cycle;
				
				\foreach \i in {1, ..., \clauseCount}	{	
					\foreach \j in {1, ..., 3} {
						\node[follower_node] (C\i\j) at (6,\i*3.5-\j) {\small $c_{\pgfmathprint{int(\clauseCount+1-\i)}}^\j$};
						\node[follower_node] (C\i) at (8,\i*3.5-2) {\small $c_{\pgfmathprint{int(\clauseCount+1-\i)}}$};
						\draw (C\i) -- (C\i\j);	
					}
					\draw (C\i1) -- (C\i2);
					\draw (C\i1) to[bend right=50] (C\i3);
					\draw (C\i2) --  (C\i3);
					\node[follower_node] (z) at (11,7) {\small $z$};
					\draw (C\i) --  (z);
				}
				
				\draw (C41.100) -- (L1*2-1);
				\draw (C42) -- (L2*2-1);
				\draw (C43) -- (L3*2);
				\draw (C31) -- (L2*2);
				\draw (C32) -- (L1*2);
				\draw (C33) -- (F1*2+1);
				\draw (C21) -- (L3*2-1);
				\draw (C22) -- (L1*2-1);
				\draw (C23) -- (F3*2);
				\draw (C11) -- (F1*2);
				\draw (C12) -- (F2*2);
				\draw (C13) -- (F3*2+1);
				
				\node at (1, 14.5) {\small \ew{M}{0}};
				\node at (1, -.5) {\small \ew{1}{M}};
				\node at (6, 14.5) {\small \ew{1}{10}};		
				\node at (8, 13.5) {\small \ew{R}{5}};		
				\node at (11, 8.5) {\small \ew{0}{1}};
			\end{scope}
		\end{tikzpicture}
	\end{center}
	\caption{$\sigmap$-hardness reduction from \bcnf to \bis for variants $(\costleader_s/\costleader_b, \costfollower_s, \optimistic/\pessimistic)$.}
	\label{fig:reduction:general}
\end{figure}

So far we have constructed the underlying graph of the instance $I'$ of \bis. 
In \Cref{fig:reduction:general}, the set of blue vertices is~$V_\ell$ and is controlled by the leader, while the set of red vertices is~$V_f$ and is controlled by the follower. 
Let $M, R\in \RR_+$ such that $M \gg R \gg \max\{n_1, n_2\}$.
Finally we conclude our construction by defining the weight functions for the leader and the follower as follows:
\begin{align*}
	(w_\ell(v), w_f(v)) &= \begin{cases}
		(M, 0) & \text{if } v \in A, \\
		(1, M) & \text{if } v \in B, \\
		(1, 10) & \text{if } v \in \set{c_i^1, c_i^2, c_i^3} \text{ for all } i \in [m], \\
		(R, 5) & \text{if } v = c_i \text{ for all } i \in [m], \\
		(0, 1) & \text{if } v = z.
	\end{cases}
\end{align*}
Note that the constructed graph is a simple undirected graph, and the construction itself can be done in polynomial time.
Based on the above construction, we can prove the following lemma.

\begin{restatable}[]{lemma}{BISCorrectnessSigmap}
	\label{lemma:sigmap2-hardness}
	In the \bis instance, when the follower has a sum type objective function, $\costfollower_s$, the leader achieves an objective value $\geq M\cdot n_1 + R$ (when considered $\costleader_s$) and value $\geq 1$ (when considered $\costleader_b$), if and only if, in the \bcnf instance, there exists a truth assignment for the variables in $X$ such that the formula $\phi$ is not satisfied for any truth assignment of the variables in $Y$.
\end{restatable}
\begin{proof}
	First, for the follower's optimal reaction $F$, observe that $z \notin F$ if and only if $c_i \in F$ for some $i \in [m]$.
	Now, we argue that, for the leader's objective function $\costleader_s$ (or $\costleader_b$), her objective value is $\geq M\cdot n_1 + R$ (or $\geq 1$) if and only if $F$ does not contain the vertex $z$.
	
	When the leader has a sum type objective function, $\costleader_s$, her optimal action $L$ always contains at least one vertex from each set of pairs $\set{a_i, \bar{a}_i}$ because of their associated weight~$M$, and at most one vertex because $a_i$ and $\bar{a}_i$ are adjacent to each other.
	There are $n_1$ many such pairs $\set{a_i, \bar{a}_i}$.
	Thus, the leader chooses $n_1$ many blue vertices in her action.
	Moreover, $z \notin F$ only when $c_i \in F$ for some $i \in [m]$.
	Since each vertex $c_i$ carries weight $R$ for the leader, clearly the leader's objective value, for objective function $\costleader_s$, is assured to be at least $M \cdot n_1 + R$ if and only if the follower with a sum type objective function does not add vertex $z$ to $F$, because it only happens if he adds at least one of the $c_i$'s in $F$. 
	When the leader has objective function $\costleader_b$, irrespective of her chosen action $L$, her objective value is $0$ if $z\in F$, and $1$ otherwise.
	
	In our construction, the vertices $c_i^1, c_i^2, c_i^3, c_i$ form a clique, so the follower, due to his sum type function, always chooses exactly one vertex from set $\set{c_i^1, c_i^2, c_i^3, c_i}$ for each $i\in [m]$.
	Due to his weight function $w_f$, he always prefers either of $c_i^1, c_i^2,$ or $c_i^3$ over $c_i$. 
	As a result, vertex $c_i$ belongs to $F$ if and only if none of $c_i^1, c_i^2,$ or $c_i^3$ can be included in $F$.
	It is crucial to note that the follower can add $c_i^j$ to $F$ only if its neighbor in $A \cup B$ is not in $L \cup F$. 
	Also recall that, according to our construction, $c_i^j$ 
	has a neighbor in $A \cup B$ which is the negation of its corresponding literal $l_i^j$.
	
	Finally, we claim that there exists a leader's action $L$ for instance $I'$ of \bis wherein, for all possible reactions $F$ of the follower, $F$ cannot contain any of the $c_i^1, c_i^2, c_i^3$ for some $i$, if and only if for the instance $I$ of \bcnf, there exists a truth assignment of $X$ variables which renders the formula $\phi$ unsatisfied for any truth assignment of the $Y$ variables.
	
	\forward In the instance $I$, let $f \colon X \to \set{0,1}$ be a truth assignment such that there is no truth assignment for $Y$ which satisfies formula $\phi$.
	Then consider $L = \set{a_i \mid f(x_i) = 1} \cup \set{\bar{a}_i \mid f(x_i) = 0}$. 
	For this action $L$ of the leader, we claim that there exists at least one $i\in [m]$ such that the follower cannot add any of $c_i^1, c_i^2, c_i^3$ into $F$.
	Assume otherwise, and there exists an optimal reaction $F'$ of the follower such that for all $i\in [m]$ at least one of $c_i^1, c_i^2, c_i^3$ is in $F'$.
	Note that $F'$ also contains exactly one vertex from each pair $\{b_i, \bar{b}_{i}\}$, since the follower's objective function is of sum type.
	Then consider the truth assignment $f' \colon Y \to \set{0,1}$ where we set $f'(y_i) = 1$ if and only if $\bar{b}_i\in F'$.
	Note that for each set $\{c_i^1, c_i^2, c_i^3\}$, at least one of their neighbors in $A\cup B$ is not in $L\cup F'$.
	Since each vertex corresponding to a literal $l^j_i$ is adjacent to a vertex corresponding to its negative counterpart in $A\cup B$, this implies the truth assignment $f$ of the variables in $X$, and the truth assignment $f'$ of the variables in $Y$ satisfy the formula $\phi$. 
	We get a contradiction to the assumption. 
	
	\backward Consider the leader's action $L$ which results in an optimal objective value $\geq M\cdot n_1+R$ (when considered $\costleader_s$) and $\geq 1$ (when considered $\costleader_b$). 
	Then consider the truth assignment $f_X \colon X \to \set{0,1}$ where $f_X(x_i) = 1$ if and only if $\bar{a}_i \in L$. 
	We claim that this is a feasible solution of the \bcnf instance $I$, i.e., $f_X$ cannot be extended to a satisfying assignment.
	Assume otherwise, and let $f_Y$ be the truth assignment to the variables in $Y$ such that $f_X$ and $f_Y$ satisfy $\phi$.
	Now consider the follower's partial reaction $P_F$ which contains the set $\set{b_i \mid f_Y(y_i) = 0} \cup \set{\bar{b}_i\mid f_Y(y_i) = 1}$ in instance $I'$. 
	We claim that the partial reaction $P_F$ of the follower can be extended to a complete reaction $F$ such that $L\cup F$ returns an optimal leader's value $< M\cdot n_1+R$ (when considered $\costleader_s$) and $< 1$ (when considered $\costleader_b$). 
	This is a contradiction.
	
	Note that, if each clause $C_i$ is satisfied with assignment $f_X,f_Y$, at least one of the literals corresponding to $c_i^1,c_i^2,$ or $c_i^3$ has truth value $1$.
	Let $c_i' \in \{c_i^1, c_i^2, c_i^3\}$ be a vertex whose corresponding literal evaluates to $1$.
	Thus, according to the construction of assignments $f_X$ and $f_Y$, the neighbor of $c_i'$ in $A\cup B$ does not belong to $L\cup P_F$.
	As $N(c_i') \cap (A\cup B) \notin L\cup P_F$,
	the set $\{c_i'\}\cup (L\cup P_F)$ is also an independent set.
	Extrapolating this to every clause, the follower will be able to add at least one of the vertices $c_i^1, c_i^2$, and $c_i^3$ to his  partial solution $P_F$ for all $i\in [m]$.
	Note that all these vertices are pairwise non-adjacent.
	So altogether they form an independent set.
	Thus none of $\{c_1, c_2, \dots, c_m\}$ can be included into the follower's reaction anymore, as for each $i\in[m]$ at least one of $c_i$'s neighbors from $\{c_i^1, c_i^2, c_i^3\}$ is already in the follower's reaction.
	As a result, the follower will include the vertex $z$ into his reaction. 
	Thus, the leader's objective function will evaluate to $< M\cdot n_1+R$ (when considered $\costleader_s$) and $< 1$ (when considered $\costleader_b$).
\end{proof}

\begin{restatable}[]{theorem}{SigmaPCompleteness}
	\label{thm:general:sigmap2}
	The decision version of the \bis problem is \sigmap-complete when the leader has bottleneck or sum objective function $(\costleader_b/\costleader_s)$, the follower has sum objective function $(\costfollower_s)$, and the follower behaves in an optimistic or pessimistic way $(\optimistic/\pessimistic)$.
\end{restatable}
\begin{proof}
	\Cref{lemma:sigmap2-hardness} establishes the correctness of our reduction when considering the leader's objective function $\costleader\in \{\costleader_s, \costleader_b\}$ and the follower's objective function $\costfollower_s$.
	Consequently, it proves that \bis is computationally at least as hard as the problem \bcnf over those cases. 
	Moreover, note that, these results hold irrespective of the behavior of the follower.
	Thus, the results hold for both the optimistic and the pessimistic setting of the follower. 
	As a consequence, we showed that \bis is \sigmap-hard when considering the leader's objective function to be $\costleader_s$ or $\costleader_b$, the follower's objective function to be $\costfollower_s$, and the follower's behavior to be $\optimistic$ or $\pessimistic$.
	
	Furthermore, it is easy to see that the problem belongs to the class \sigmap\ of the polynomial hierarchy as it can be formulated using an existential quantifier followed by a universal quantifier as follows:
	\[
		\exists L\subseteq V_\ell, F\subseteq V_f \quad \forall F'\subseteq V_f \quad
		\costfollower_s(L\cup F) \geq \costfollower_s(L\cup F')
		\text{ and } \costleader(L\cup F) \geq k .
	\]
	where $\costleader$ is the leader's objective function; $\costleader \in \{\costleader_s,\costleader_b\}$.
	This way for a given certificate $(L,F)$ consisting of a leader's action and a follower's reaction, we need to verify whether $F$ is, in fact, an optimal reaction of the follower to the given action $L$; we do this with the help of the condition $\costfollower_s(L\cup F) \geq \costfollower_s(L\cup F')$ for all $F'\subseteq V_f$.
	We additionally need to make sure that the final objective value of the leader is at least $k$; this is achieved using the condition $\costleader(L\cup F) \geq k$.
	Moreover, the optimistic or the pessimistic behavior of the follower can be captured by updating the follower's weight function such that it also captures the essence of the leader's weight function, i.e., by assuming $w_f(v)\leftarrow w_f(v) + \epsilon\cdot w_\ell(v)$ in the optimistic setting and $w_f(v) \leftarrow w_f(v) - \epsilon\cdot w_\ell(v)$ in the pessimistic setting.
	Thus the above formulation evaluates to true if and only if, for a given instance of \bis, there exists a leader's action $L$ which can be extended by a follower's optimal reaction $F$ such that $\costleader(L\cup F)\geq k$.
\end{proof}

Now, we consider the cases where both the leader and the follower have bottleneck objective functions, i.e., we consider variants $(\costleader_b, \costfollower_b, \optimistic/\pessimistic)$. 

\begin{restatable}[]{theorem}{BISGeneralP}\label{thm:general:P}
	The \bis problem is polynomial-time solvable when both the leader and the follower have bottleneck objective functions and the follower behaves optimistically, i.e., when considering $(\costleader_b, \costfollower_b, \optimistic)$. 
\end{restatable}
\begin{proof}
	Since the follower behaves optimistically, given a leader's non-empty action $L \subseteq V_\ell, L\neq \emptyset$, the follower can be assumed to react by selecting an empty set. 
	This is because the follower does not improve his objective function $\costfollower_b$ by any other choice of set~$F$ but might only reduce the leader's objective value by doing so. 
	Thus $F_{L\neq \emptyset} = \emptyset$, where $F_{L\neq \emptyset}$ denotes the optimal reaction of the follower for the leader's action $L\neq \emptyset$.
	However, when $L=\emptyset$, the follower reacts by choosing a single element $v\in V_f$ which optimizes the following objective function:
	\begin{align*}
		\max_{v\in F^\star}\quad  & w_\ell(v)\\
		&\text{where } F^\star = \argmax_{v\in V_f}\quad w_f(v)
	\end{align*}
	In other words, the follower reaction $F_{L=\emptyset}$ contains a vertex which maximizes his own objective value,
	and when there are multiple optimal reactions, he chooses one which additionally maximizes $w_\ell(v)$.
	
	For the case $(\costleader_b, \costfollower_b, \optimistic)$, if $L\neq \emptyset$, then it is safe to assume $L\subseteq \big \{v'\mid v'\in \argmax_{v\in V_\ell} w_\ell(v) \big \}$ because the follower's optimal reaction is $F_{L\neq \emptyset}=\emptyset$, and thus $\max_{v\in V_\ell} w_\ell(v)$ is the best objective value that the leader can achieve.
	
	\sloppypar{In consideration of the above cases, the leader strategically chooses either a vertex from $\argmax_{v\in V_\ell} w_\ell(v)$ if $w_\ell(F_{L=\emptyset}) \leq \max_{v\in V_\ell} w_\ell(v)$, or $L = \emptyset$ if $w_\ell(F_{L=\emptyset}) > \max_{v\in V_\ell} w_\ell(v)$.
	}

	Calculating both $\argmax_{v\in V_\ell} w_\ell(v)$ and $F_{L=\emptyset}$ can be done in polynomial time.
	Thus, the $(\costleader_b, \costfollower_b, \optimistic)$ variant of \bis is tractable. 
\end{proof}

\begin{restatable}[]{theorem}{BISCbDbPNPhard}\label{thm:c_bd_bp}
	The decision version of the \bis problem is \np-complete when both leader and follower have bottleneck objective functions and the follower behaves pessimistically, i.e., when considering $(\costleader_b, \costfollower_b, \pessimistic)$.
\end{restatable}
\begin{proof}
	In this case, given a leader's action $L\subseteq V_\ell$, the follower's optimal reaction is as follows:
	\begin{itemize}
		\item When $L=\emptyset$, the follower's reaction is 
		\[
		F_L = \min_{v\in F^\star}~ w_\ell(v), \qquad  \text{where } F^\star = \argmax_{v \in V_f}~ w_f(v)
		\]
		The follower picks---among all maximum weighted vertices according to the weight function $w_f$---one which has the minimum weight according to the weight function $w_\ell$.
		\item When $L\neq \emptyset$, the follower reacts by choosing
		\[
		F_L = \argmin_{\substack{v \in V_f\setminus N(L)\\ w_f(v)\geq \costfollower_b(L)}} \quad w_\ell(v),
		\]		
	\end{itemize}
	if no such $v\in V_f\setminus N(L)$ exists, then $F_L=\emptyset$
	A pessimistic follower always tries to sabotage the objective function of the leader as long as he does not reduce his own objective value in the process.
	Given a leader's action $L\subseteq V_\ell$ as a certificate, the optimal reaction of the follower $F_L$ can be calculated in polynomial time, and one can verify the leader's final objective value $\costleader_b(L \cup F_L)$. 
	Thus the problem belongs to the class \np.
	
	In order to show hardness, we give a reduction from the classical \textsc{Vertex Cover} problem to the decision version of \bis.
	In \textsc{Vertex Cover}, given an undirected graph $G = (V, E)$ and a positive integer $k$, we are asked to find a subset $S \subseteq V$ of at most $k$ vertices such that the removal of this set renders the graph edgeless, i.e., $G[V \setminus S]$ is an empty graph.
	
	\begin{figure}[t]
		\begin{center}
			\begin{tikzpicture}[scale=0.8]
				\def\height{3}
				\tikzset{
					leader_node/.style={circle, draw=black, fill=leader_color, minimum size=6mm, inner sep=0pt},
					follower_node/.style={circle, draw=black, fill=follower_color, minimum size=6mm, inner sep=0pt}
				}
					
				\draw[black!40, dashed, rounded corners] (-0.5,\height+0.5) -- (3.5,\height+0.5) -- (3.5,\height-0.5) -- (-0.5,\height-0.5) -- cycle;
				\draw[black!40, dashed, rounded corners] (4.5,\height+0.5) -- (8.5,\height+0.5) -- (8.5,\height-0.5) -- (4.5,\height-0.5) -- cycle;
				\draw[black!40, dashed, rounded corners] (9.5,\height+0.5) -- (13.5,\height+0.5) -- (13.5,\height-0.5) -- (9.5,\height-0.5) -- cycle;
				\draw[black!40, dashed, rounded corners] (4.5,0.5) -- (9.5,0.5) -- (9.5,-0.5) -- (4.5,-0.5) -- cycle;
				
				\node[leader_node] (v11) at (0,\height) {$v_1^1$};
				\node[leader_node] (v12) at (1,\height) {$v_1^2$};
				\node () at (2,\height) {$\cdots$};
				\node[leader_node] (v1k) at (3,\height) {$v_1^k$};
				
				\node[leader_node] (v21) at (5,\height) {$v_2^1$};
				\node[leader_node] (v22) at (6,\height) {$v_2^2$};
				\node () at (7,\height) {$\cdots$};
				\node[leader_node] (v2k) at (8,\height) {$v_2^k$};
				
				\node () at (9,\height) {$\cdots$};
				
				\node[leader_node] (v31) at (10,\height) {$v_n^1$};
				\node[leader_node] (v32) at (11,\height) {$v_n^2$};
				\node () at (12,\height) {$\cdots$};
				\node[leader_node] (v3k) at (13,\height) {$v_n^k$};
				
				\node[follower_node] (ve1) at (5,0) {$v_{e_1}$};
				\node[follower_node] (ve2) at (6,0) {$v_{e_2}$};
				\node[follower_node] (ve3) at (7,0) {$v_{e_3}$};
				\node () at (8,0) {$\cdots$};
				\node[follower_node] (vem) at (9,0) {$v_{e_m}$};
				
				\node at (-1.2, \height) {\ew{1}{M}};
				\node at (3.8, 0) {\ew{0}{M}};
				
				\draw (v11) to[bend left=40] (v21);
				\draw (v11) to[bend left=40] (v31);
				\draw (v21) to[bend left=40] (v31);
				\draw (v12) to[bend left=40] (v22);
				\draw (v12) to[bend left=40] (v32);
				\draw (v22) to[bend left=40] (v32);
				\draw (v1k) to[bend left=40] (v2k);
				\draw (v1k) to[bend left=40] (v3k);
				\draw (v2k) to[bend left=40] (v3k);
				
				\draw (v11) --(ve1);
				\draw (v12) --(ve1);
				\draw (v1k) --(ve1);
				\draw (v11) --(ve2);
				\draw (v12) --(ve2);
				\draw (v1k) --(ve2);
				\draw (v21) --(ve2);
				\draw (v22) --(ve2);
				\draw (v2k) --(ve2);			
				\draw (v21) --(ve3);
				\draw (v22) --(ve3);
				\draw (v2k) --(ve3);			
				\draw (v31) --(vem);
				\draw (v32) --(vem);
				\draw (v3k) --(vem);
			\end{tikzpicture}
		\end{center}
		\caption{\np-hardness reduction for the \bis problem when considering $(\costleader_b, \costfollower_b, \pessimistic)$. The blue vertices form the set $V_\ell$ whereas the red vertices form the set~$V_f$.}
		\label{fig:general:c_bd_bp}
	\end{figure}
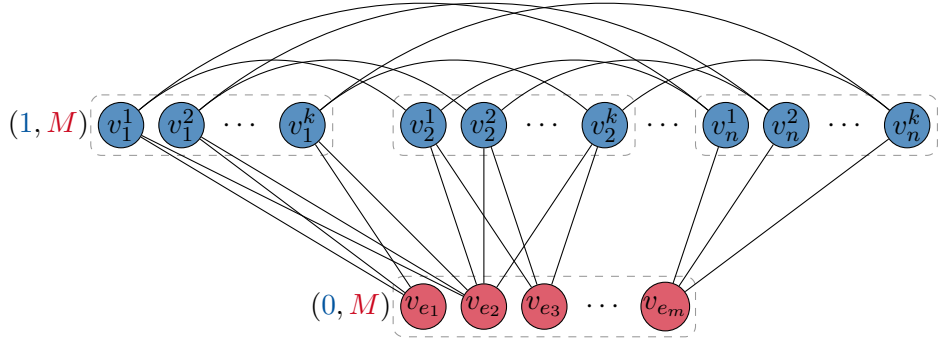
	
	Let $I=(G=(V,E), k)$ be an instance of the \textsc{Vertex Cover} problem.
	Then we construct a new instance $I' = (G' = (\V, E'), k', w_\ell, w_f )$ of the decision version of \bis where both leader and follower have bottleneck type objective functions and the follower behaves pessimistically.
	Here, $k'$ is the threshold in the decision version of \bis. 
	We give the construction as follows (also see \Cref{fig:general:c_bd_bp}):
	\begin{itemize}
		\item For each $v \in V$, add $k$ copies of $v$ to $V_\ell$; let $v^1, v^2, \ldots, v^k \in V_\ell$ for $i\in [k]$. 
		Define $V^i = \set{v^i \mid v \in V}$ for all $i\in[k]$.
		\item For each $e \in E$, add a vertex $v_e \in V_f$.
		\item If an edge $e$ is incident to $v$, add edges $\{v^i, v_e\}\in E'$ for all $i\in[k]$.
		\item Make each $G[V^i]$ a clique, so add edges between each pair of vertices in $V^i$.
		\item Let $k' = 1$.
	\end{itemize}
	\noindent Finally we conclude our construction by defining the following weight functions: 
	\begin{align*}
		w_\ell(v) &= \begin{cases}
			1 & \text{if } v \in V_\ell, \\
			0 & \text{if } v \in V_f,
		\end{cases} &
		w_f(v) &= \begin{cases}
			M & \text{if } v \in V_\ell, \\
			M & \text{if } v \in V_f,
		\end{cases}
	\end{align*}
	where $M\gg 0$.
	
	\noindent\textbf{Correctness:}
	Now, we show that $I$ is an yes-instance of the \textsc{Vertex Cover} problem if and only if $I'$ is a yes-instance of \bis when considered $(\costleader_b, \costfollower_b, \pessimistic)$.
	
	\forward In the instance $I$ of the {\sc Vertex Cover} problem, assume that there is a subset $S \subseteq V$ with $|S| \leq k$ such that $G[V \setminus S]$ does not contain any edges.
	Then we claim that there is a leader's action $L \subseteq V_\ell$ in $I'$ which yields an objective value $1$ for the leader.
	From each set $V^i$, $i \in [k]$, we pick one vertex corresponding to a vertex in $S$.
	Let the resultant set be the leader's action $L$.
	We have at most $k$ vertices in $S$, so it can be guaranteed that $L$ is an independent set.
	Moreover, there are at most $k$ vertices in $S$ and we have $k$ many different sets $V^i$ each for $i\in [k]$; we can cover one copy of each vertex in $S$ from the sets $V^i$.
	Given this choice $L$ of the leader, the follower cannot pick any vertex in the solution because $V_f\setminus N(L) = \emptyset$, as every edge $e\in E$ was covered by a vertex in~$S$.		
	Since $F_L=\emptyset$, the leader's objective function is simply $\costleader_b(L\cup F_L)=\costleader_b(L) = 1$.
	
	\backward Now, if there exists a leader's action $L \subseteq V_\ell$ for which she achieves the objective value at least 1, then it is clear that $V_f \subseteq N(L)$, since the follower was not able to sabotage the leader's objective function by picking a vertex from set~$V_f$. 
	Moreover, since $L$ is an independent set, $|L| \leq k$ because the leader can choose at most one vertex from each of the sets $V^i$, $i\in [k]$.
	Consider the corresponding vertices in $V(G)$ whose copies were chosen in $L$. 
	Let us denote this set by $S$; $|S|\leq k$.
	Note that $S$ covers all edges in $E(G)$, since $V_f \subseteq N(L)$.
	Thus, $S$ is a vertex cover of $G$.
\end{proof}

Next, we resolve the remaining cases $(\costleader_s, \costfollower_b, \optimistic/\pessimistic)$ of \bis on general graphs.

\begin{restatable}[]{theorem}{BISCsDbOPNPcomplete}\label{thm:bis:gen:np-hard:c_sd_bp/o}
	The decision version of the \bis problem is \np-complete when the leader's objective function is of sum type, the follower's objective function is of bottleneck type, and the follower behaves optimistically or pessimistically, i.e., when considering $(\costleader_s, \costfollower_b, \optimistic)$ and $(\costleader_s, \costfollower_b, \pessimistic)$.
\end{restatable}
\begin{proof}
	First, we prove the \np-hardness for both these cases.
	We give a naive reduction from the \independentset problem which works for both the variants.
	It is well-known that \independentset is \np-complete~\cite{karp2009reducibility}.
	Consider the underlying graph $G = (V,E)$ of an instance $I$ of  \independentset.
	Take a replica of this graph to construct an instance $I'$ of \bis.
	Let the vertex sets controlled by the leader and the follower be $V_\ell = V$ and $V_f = \emptyset$, respectively.
	Moreover, consider the weight functions of the leader and the follower to be $w_\ell, w_f \colon V \to \{1\}$.
	
	Note that the follower always has an empty reaction irrespective of his behavior because $V_f = \emptyset$ implies he does not have a choice.
	So the leader does not have to care about the follower's reaction.
	Hence, for the variant $(\costleader_s, \costfollower_b, \pessimistic)$ and $(\costleader_s, \costfollower_b, \optimistic)$, the \bis problem boils down to solving for the maximum weighted independent set in the graph according to function $w_\ell$.
	This is equivalent to solving the original instance $I$ of the \independentset problem.
	Thus our reduction is sound.
	The variants $(\costleader_s, \costfollower_b, \optimistic)$ and $(\costleader_s, \costfollower_b, \pessimistic)$ of \bis are \np-hard. 
	
	Now, for each of the cases, we individually prove it's containment in \np.
	
	First, we prove that \bis variant $(\costleader_s, \costfollower_b, \pessimistic)$ is in \np.
	First, we show the containment in the class \np.
	The certificate is just the leader's solution $L$.
	If $L = \emptyset$, then a pessimistic follower selects a vertex that is best for him and worst for the leader which can be computed using:
	\begin{align*}
		\min_{v\in F^\star}\quad  & w_\ell(v)\\
		&\text{where } F^\star = \argmax_{v\in V_f}~ w_f(v).
	\end{align*}
	If $L \not = \emptyset$, then any additional vertex by the follower can neither improve the follower's objective value nor make the leader's value worse, hence the follower selects nothing.
	Hence, for any given certificate $L\subseteq V_\ell$, one can calculate the follower's optimal reaction, and consequently verify the leader's final objective value.
	Thus, the problem is in \np.		
	
	Next, we prove that \bis variant $(\costleader_s, \costfollower_b, \optimistic)$ is in \np.
	To show containment in \np, we show that there is a valid polynomial size certificate which can be verified in polynomial time.
	Let the certificate be made up of both the leader's and the follower's chosen set of variables $L \cup F$.
	To check whether $L\cup F$ is a valid solution of \bis, we need to check if $L\cup F$ forms an independent set.
	However, along with that, we also need to prove that 
	the follower's optimal reaction $F_L$ to the leader's action $L$ evaluates the objective function of the leader to a value that is at most as much as the solution $L\cup F$ does.
	Once these conditions are verified, it is easy to calculate the objective value of the leader, so calculate $\sum_{v\in L\cup F} w_\ell(v)$ to verify whether the given instance is a yes-instance of variant $(\costleader_s, \costfollower_b, \optimistic)$ of \bis.
	Verifying if $L\cup F$ is an independent set is trivial in polynomial time.
	Consider the induced subgraph over the vertex set $L\cup F$ and check whether there is an edge in it.
	If $L\neq \emptyset$, then the follower's objective value cannot be greater than $\costfollower_b(L)$.
	In this case, the follower's objective value can never improve by choosing more vertices, it is determined entirely by the leader's solution.
	Thus, in addition to $L\cup F$ being an independent set, if $\costfollower_b(L\cup F) = \costfollower_b(L)$ for $L\neq \emptyset$, then $F$ belongs to the set of the follower's possible reactions to $L$.
	Now, for the part when $L=\emptyset$, the follower's reaction $F$ needs to satisfy the following conditions:
	\begin{enumerate}
		\item[(1)] $L\cup F$ is an independent set in the graph, which is equivalent to $F$ being an independent set, 
		\item[(2)] $F \subseteq \argmax \{w_f(v) \mid v\in V_f\}$ because the follower wants to maximize the minimum weighted vertex in the solution, thus he would only chose from the set of vertices which gives him the highest weight possible.
	\end{enumerate}
	It is easy to check whether conditions (1) and (2) hold.
	Moreover, since the follower behaves optimistically towards the leader, it suffices to have one witness (for example, the set $F$ here) to show that, the leader's objective value evaluates to at least $\sum_{v\in L\cup F} w_\ell(v)$.	This is because the optimistic follower would help maximize the objective function of the leader.
	Thus, he will always choose a set $F \subseteq \argmax \{w_f(v) \mid v\in V_f\}$ which works the best for the leader.
	Thus, the variant $(\costleader_s, \costfollower_b, \optimistic)$ of \bis belongs to the class \np.		
\end{proof}

\subsection{Complexity on Bipartite Graphs}
\label{sec:bis:bipartite}

In this section, we study \bis on bipartite graphs. 
The maximum weight \independentset problem is polynomial-time solvable on bipartite graphs~\cite{Edmonds}. 
However, when considering bilevel variants of this problem, we get a variety of results on different variants.
Please refer to \Cref{subsec:bo:contribution} for an overview of our results given in \Cref{results:overview}.

We first consider the variants $(\costleader_s, \costfollower_s, \optimistic/ \pessimistic)$.

\subsubsection{Both decision makers have sum type objective functions}
\label{subsubsec:bis:bipartite:sum-type}

According to the upcoming result, we not only show \np-completeness on bipartite graphs, but we show it on an even restricted planar bipartite graphs.

\begin{restatable}[]{theorem}{BipartiteCsDsOPNPComplete}
	\label{thm:bipartite:c_sd_so/p}
	Let $G = (V, E)$ be a planar bipartite graph where $V=\V$, along with two weight functions $w_\ell, w_f\colon V\to \RR_+$ such that $V_\ell, w_\ell$ correspond to the leader and $V_f, w_f$ correspond to the follower.
	The decision version of the \bis problem is \np-complete when the leader has a sum objective function $(\costleader_s)$, the follower has a sum objective function $(\costfollower_s)$, and the follower behaves in an optimistic or pessimistic way $(\optimistic/ \pessimistic)$, i.e., when considering $(\costleader_s, \costfollower_s, \optimistic/\pessimistic)$.
\end{restatable}
\begin{proof}
	For both variants of the problem $(\costleader_s, \costfollower_s, \optimistic/\pessimistic)$, the problem lies in \np because a valid certificate is the leader's solution. 
	The follower's solution is just a weighted independent set problem on the remaining graph which is computable in polynomial time as the graph is bipartite. 
	This consequently shows \bis on planar bipartite graphs is also in \np.
	The optimistic or pessimistic behavior of follower can be captured by using the updated weight function of the follower, i.e., by assuming $w_f(v)\leftarrow w_f(v) + \epsilon\cdot w_\ell(v)$ in the optimistic setting and $w_f(v) \leftarrow w_f(v) - \epsilon\cdot w_\ell(v)$ in the pessimistic setting.
	
	Moreover, to prove the \np-hardness of the problems, we give a reduction from the planar vertex cover problem which has been shown to be \np-hard~\cite{Lichtenstein1982PlanarFA}.
	Given a planar graph and a positive integer $k$, the planar vertex cover problem asks if there is a vertex cover of size at most $k$.
	
	\textbf{Construction:} Let $I = (G = (V, E), k)$ be an instance of the planar vertex cover problem.
	Let $V = \set{v_1, v_2, \ldots, v_n}$ and $E = \set{e_1, e_2, \ldots, e_m}$. 
	Now, we construct an instance $I' = (G'= (V_\ell \cup V_f, E'), k', w_\ell, w_f)$ of the decision version of \bis as follows (also see \Cref{fig:reduction:bipartite}):
	\begin{itemize}
		\item For each $v_i \in V$, add vertices $a_i, a_i'$ in $V_\ell$. 
		Define $A \coloneqq \set{a_1, \ldots, a_n}$ and $A' \coloneqq \set{a_1', \ldots, a_n'}$. 
		\item For each $e_i \in E$, add vertices $b_i, b_i'$ in $V_f$. 
		Define $B \coloneqq \set{b_1, \ldots, b_m}$ and $B' \coloneqq \set{b_1', \ldots, b_m'}$.
		\item For each $i \in [n]$, add an edge $\{a_i, a_i'\}$ to $E'$. 
		For each $j \in [m]$, add an edge $\{b_j,b_j'\}$ to $E'$.
		\item Add an edge $\{a_i, b_j\}$ to $E'$ if and only if edge $e_j$ is incident to vertex $v_i$ in $E$.
	\end{itemize}
	
	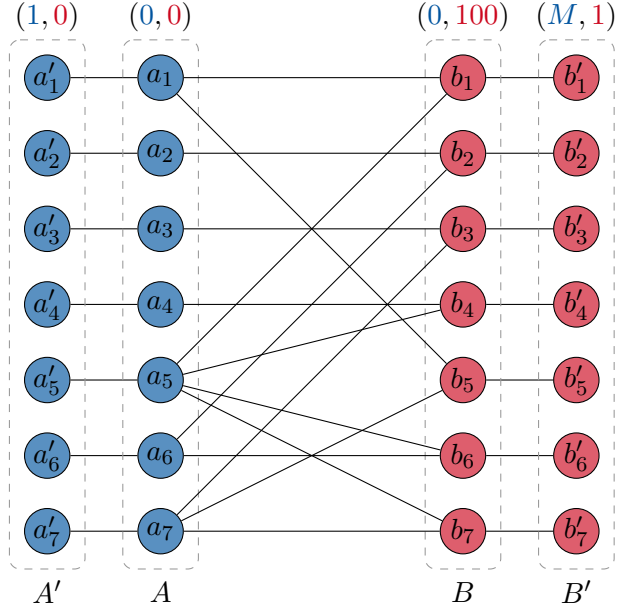
\begin{figure}[t]
		\begin{center}
			\begin{tikzpicture}
				\def\vertexCount{7}
				\tikzset{
					leader_node/.style={circle, draw=black, fill=leader_color, minimum size=6mm, inner sep=0pt},
					follower_node/.style={circle, draw=black, fill=follower_color, minimum size=6mm, inner sep=0pt}
				}
				
				\foreach \i in {1,...,\vertexCount} {
					\node[leader_node] (L\i) at (0,\i) {$a_{\pgfmathprint{int(8-\i)}}$};
					\node[leader_node] (L'\i) at (-1.5,\i) {$a'_{\pgfmathprint{int(8-\i)}}$};
					\node[follower_node] (F\i) at (4,\i) {$b_{\pgfmathprint{int(8-\i}}$};
					\node[follower_node] (F'\i) at (5.5,\i) {$b'_{\pgfmathprint{int(8-\i}}$};
					\draw (L\i) -- (L'\i);
					\draw (F\i) -- (F'\i);
				}
				
				\node at (-1.5, 7.8) {\ew{1}{0}};
				\node at (0, 7.8) {\ew{0}{0}};
				\node at (4, 7.8) {\ew{0}{100}};
				\node at (5.5, 7.8) {\ew{M}{1}};
								
				\draw[black!40, dashed, rounded corners] (-2,7.5) -- (-1,7.5) --(-1,0.5)-- (-2,0.5) -- cycle;
				\draw[black!40, dashed, rounded corners] (-0.5,7.5) -- (0.5,7.5) --(0.5,0.5)-- (-0.5,0.5) -- cycle;
				\draw[black!40, dashed, rounded corners] (3.5,7.5) -- (4.5,7.5) --(4.5,0.5)-- (3.5,0.5) -- cycle;
				\draw[black!40, dashed, rounded corners] (5,7.5) -- (6,7.5) --(6,0.5)-- (5,0.5) -- cycle;
				
				\node at (-1.5, 0.2) {$A'$};
				\node at (0, 0.2) {$A$};
				\node at (4, 0.2) {$B$};		
				\node at (5.5, 0.2) {$B'$};	
				
				\draw (F1) -- (L1);
				\draw (F1) -- (L3);
				\draw (F2) -- (L2);
				\draw (F2) -- (L3);
				\draw (F3) -- (L1);
				\draw (F3) -- (L7);
				\draw (F4) -- (L3);
				\draw (F4) -- (L4);
				\draw (F5) -- (L1);
				\draw (F5) -- (L5);
				\draw (F6) -- (L2);
				\draw (F6) -- (L6);
				\draw (F7) -- (L3);
				\draw (F7) -- (L7);
			\end{tikzpicture}
		\end{center}
		\caption{\np-hardness reduction for the \bis problem on planar bipartite graphs when considering $(\costleader_s, \costfollower_s, \optimistic/\pessimistic)$.
		}
		\label{fig:reduction:bipartite}
	\end{figure}
	
	As \Cref{fig:reduction:bipartite} shows, the reduced instance is clearly a bipartite graph, however, note that, the coloring depicted in the figure is not according to the bipartition.
	Now, we argue that the constructed instance is additionally a planar graph.
	This holds because we can use a planar embedding of the original graph as a backbone to obtain a  planar embedding of the constructed graph: Consider a planar embedding of graph $G$, replace each $v_i$ by~$a_i$, delete each edge $e_j$ and instead add a new vertex $b_j$ such that $b_j$ is embedded at the midpoint of edge $e_j$.
	Now add all edges $\{a_i,b_j\}$; due to the planar embedding of $G$ the resulting graph is still planar.
	The additional vertices $a_i'$ and $b_j'$ along with the edges $\{a_i,a_i'\}$ and $\{b_j,b_j'\}$ can be appended to vertices $a_i$ and $b_j$, respectively, while preserving the current planarity of the graph.
	
	Finally we conclude our construction by defining the following weight functions:
	\begin{align*}
		w_\ell(v) &= \begin{cases}
			0 & \text{if } v \in A \cup B, \\
			1 & \text{if } v \in A', \\
			M & \text{if } v \in B',
		\end{cases} &
		w_f(v) &= \begin{cases}
			0 & \text{if } v \in A \cup A', \\
			100 & \text{if } v \in B, \\
			1 & \text{if } v \in B',
		\end{cases}
	\end{align*}
	where $M \in \NN, 100 \ll M$ is a sufficiently large number.
	
	Now it remains to argue that, for all the considered cases, $I$ is a yes-instance of the planar vertex cover problem if and only if $I'$ is a yes-instance of \bis.
	
	\textbf{Correctness of our construction:}
	We show that, in the \bis instance, the leader achieves value at least $k'\coloneqq (mM + n-k)$ if and only if the initial instance of the planar vertex cover problem has a vertex cover of size at most $k$.
	\backward Let $S\subseteq V$ be a vertex cover of graph $G = (V, E)$ of size at most $k$.
	Then the leader's action $L = \set{a_i \in A\mid v_i \in S} \cup \set{a_i' \in A'\mid v_i \notin S}$ results in an objective value of $\geq (mM + n - k)$ for her: Clearly, $L$ is an independent set.
	Note that, for each vertex $b_i \in B$, $b_i$ has a neighbor in $L$ because, in the planar vertex cover instance, the set~$S$ covers all the edges of the graph.
	Thus, given $L$, the follower cannot pick a vertex from the set~$B$ in his reaction~$F_L$.
	Moreover, the optimal follower's reaction is $F_L = B'$, independent of whether he is optimistic or pessimistic.
	Then $\costleader_s(L\cup F_L) = mM + |\set{a_i' \in A'\mid v_i \notin S}| \geq mM + n-k$.
	This holds for both the optimistic and pessimistic settings, since the follower's reaction is unique.
	
	\forward Consider that \bis contains a leader's action $L$ which results in an objective value at least $mM + n-k$.
	This can be achieved if and only if the follower chooses $F \supseteq B'$.
	Also, note that, if both $b_i$ and $b_i'$ are available then the follower will always pick $b_i$ over $b_i'$ because it gives him a better objective value.
	Thus, for each $b_i\in B$, $L$ must contain a vertex which is adjacent to $b_i$.
	Let $L_A = L \cap A$ be such that $B\subseteq N(L_A)$.
	Note that the corresponding vertex set $S' = \set{v_i\in V\mid a_i\in L_A}$ is a vertex cover of the original instance $I$ of the planar vertex cover problem.
	Now, it only remains to prove that $\abs{L_A} \leq k$.
	Suppose the contrary, $\abs{L_A} > k$, then $|L \cap A'| < n-k$.
	Thus, $\costleader_s(L\cup F_L) < (mM + n - k)$, a contradiction.
	Hence, $\abs{L_A} \leq k$.
	
	Thus the cases $(\costleader_s, \costfollower_s, \optimistic/\pessimistic)$ are \np-complete.	
\end{proof}

\subsubsection{Leader with sum type and follower with bottleneck type objective function }
\label{subsubsec:vis:bipartite:leader-sum:follower-bottleneck}
In this section, we restrict to the variants $(\costleader_s, \costfollower_b, \optimistic)$ and $(\costleader_s, \costfollower_b, \pessimistic)$ on bipartite graphs. 
We prove that \bis is polynomial-time solvable for both of these variants.

\begin{restatable}[]{theorem}{BipartiteCsDbOPoly}
	\label{thm:bipartite:c_sd_bo}
	For a bipartite graph $G = (V, E)$, \bis is solvable in polynomial time when considering the leader's objective function of sum type, the follower's objective function of bottleneck type and the follower behaves in an optimistic way, i.e., case $(\costleader_s, \costfollower_b, \optimistic)$.
\end{restatable}
\begin{proof}
	For $L=\emptyset$, $F_{L=\emptyset}$ consists of the max-weighted independent set with respect to the leader's weight function $w_\ell$, on the graph $G[F^\star]$, where $F^\star = \argmax_{v\in V_f} w_f(v)$, which is computable in polynomial-time as we are dealing with bipartite graphs.
	
	Observe that, given a non-empty leader's action $L\subseteq V_\ell$, the follower's objective value cannot be greater than $\costfollower_b(L)$.
	Moreover, since the follower is optimistic, he will always react with the maximum weighted independent set (with respect to the weight function $w_\ell$) on the graph $G[F^\star]$, where $F^\star = \set{v \in V_f \setminus N(L) \mid w_f(v) \geq \costfollower_b(L)}$.
	In this case, the leader's optimal action can be calculated using Algorithm \ref{alg:bipartite:(c_sd_bo)}.
	\begin{algorithm}[]
		\SetAlgoRefName{$1$}
		\KwIn{A bipartite graph $G = (V,E)$, and the weight functions $w_\ell, w_f\colon V \to \RR_+$.}
		\KwOut{An optimal leader's action $L_{opt}$ for the case $(\costleader_s, \costfollower_b, \optimistic)$.}
		Calculate the follower's reaction $F_{L=\emptyset}$.\\
		$val \coloneqq w_\ell{(F_{L=\emptyset})}$\\
		$L_{opt} \coloneqq \emptyset$\\
		$R \coloneqq V_\ell$\\
		\While{$R \neq \emptyset$}{
			Select $v^\star \in R$  \tcp*{$w_f(v^\star)$ sets the threshold for $\costfollower_b(L)$ value}
			$V' = \set{v' \in V \mid w_f(v')\geq w_f(v^\star), v'\notin N[v^\star]}$\\
			Calculate the maximum weighted independent set $S_{V'}$ of $G[V']$ according to $w_\ell$.\\
			\If{$w_\ell(v^\star) + \sum_{v\in S_{V'}}w_\ell(v) > val$}{
				$val \coloneqq w_\ell(v^\star) + \sum_{v\in S_{V'}}w_\ell(v)$\\
				$L_{opt} \coloneqq \set{v^\star}\cup (S_{V'}\cap V_\ell)$\\
			}
			$R \coloneqq R \setminus \{v^\star\}$\\
		}
		\Return $L_{opt}$.
		\caption{$(\costleader_s, \costfollower_b, \optimistic)$ on bipartite graphs}
		\label{alg:bipartite:(c_sd_bo)}
	\end{algorithm} 
	
	\emph{Description of Algorithm \ref{alg:bipartite:(c_sd_bo)}:} When $L\neq \emptyset$, the follower's objective value has a lower bound; it cannot be greater than $\costfollower_b(L)$.
	In fact, it suffices to consider that only one vertex from $L$ actually enforces the bound $\costfollower_b(L)$.
	Hence, we enumerate over all possible vertices $v^\star\in V_\ell$ that can influence the value of the lower bound $\costfollower_b(L)$ for the follower.
	Once, $v^\star\in V_\ell$ is fixed, the leader considers the set of vertices $V' = \set{v' \in V \mid w_f(v')\geq w_f(v^\star), v'\notin N[v^\star]}$, and calculates the maximum weighted independent set~$S_{V'}$ in the subgraph $G[V']$ according to the weight function $w_\ell$.
	The maximum weighted independent set obtained this way, over all the vertices $v^\star$ considered, gives us a leader's optimal action. 
	
	\emph{Correctness of Algorithm \ref{alg:bipartite:(c_sd_bo)}:} 
	A leader's action $L\subseteq V_\ell$ sets a threshold for the optimistic follower's reaction. 
	The follower chooses a reaction $F_L$ such that $\min_{v\in F_L} w_f(v) \geq \costfollower_b(L)$.
	We go over all possible threshold values $\costfollower_b(L)$ that a leader's action can impose on the follower.
	And for each one of them, we calculate an optimal solution.
	Let $S\subseteq V$ be the set which returns the best possible value for the leader among all the enumerated cases.
	Then the leader's optimal action is $S\cap V_\ell$.
	This proves the correctness of our algorithm.
	
	Moreover, there are at most polynomial many steps which individually take at most polynomial time.
	Thus, Algorithm \ref{alg:bipartite:(c_sd_bo)} runs in polynomial time.
\end{proof}

\begin{restatable}[]{theorem}{BipartiteCsDbPPoly}\label{thm:bipartite:c_sd_bp}
	Given a bipartite graph $G = (V, E)$, \bis is solvable in polynomial time when considering the leader's objective function of sum type, the follower's objective function of bottleneck type, and the follower behaves in a pessimistic way, i.e., case $(\costleader_s, \costfollower_b, \pessimistic)$.
\end{restatable}
\begin{proof}
	If $L=\emptyset$, then the follower's optimal reaction is to chose from the vertices which maximizes his objective value the one with the least $w_\ell$ value.
	Otherwise, if $L\neq \emptyset$, the pessimistic follower will not choose any vertex as it will only increase the leader's objective value and not his own.
	Thus, the leader accordingly chooses either $L=\emptyset$ or the maximum weighted independent set, with respect to $w_\ell$, in $V_\ell$.
	The latter can be calculated in polynomial time~\cite{Edmonds}.		
\end{proof}

The results obtained in \Cref{subsubsec:bis:bipartite:sum-type} and \Cref{subsubsec:vis:bipartite:leader-sum:follower-bottleneck} demonstrate a decrease in the computational complexity of \bis on bipartite graphs compared to the results on general graphs.
Naturally, this is because the underlying \independentset problem is easier on bipartite graphs than on general graphs.
However, when considering the variant $(\costleader_b, \costfollower_b, \pessimistic)$ on bipartite graphs in the next section, we retain that the problem is \np-complete on bipartite graphs similar to \Cref{thm:c_bd_bp} on general graphs.
In the next section, we show that, with slight changes, the same construction works for proving \np-hardness for the variants $(\costleader_b, \costfollower_s, \optimistic/\pessimistic)$.

\subsubsection{The remaining variants on bipartite graphs}
\label{subsubsec:bis:bipartite:remaining}

Next we consider the cases $(\costleader_b, \costfollower_b, \pessimistic)$ and $(\costleader_b, \costfollower_s, \optimistic/\pessimistic)$ on bipartite graphs, and prove the \np-completeness of \bis.
We modify the \np-hardness construction used in \Cref{thm:c_bd_bp} such that the reduced instances map to bipartite graphs.

\begin{restatable}[]{theorem}{BipartiteCbDsOPNPComplete}\label{thm:bipartite:c_bd_bp:c_bd_so/p}
	The decision version of the \bis problem is \np-complete on bipartite graphs when considering any of the variants $(\costleader_b, \costfollower_b, \pessimistic)$ and $(\costleader_b, \costfollower_s, \optimistic/\pessimistic)$.
	\vspace{-.5ex}
\end{restatable}
\begin{proof}
	The \np-hardness proof in \Cref{thm:c_bd_bp} is for the variant $(\costleader_b, \costfollower_b, \pessimistic)$ on general graphs wherein the constructed graph contains odd cycles.
	To show \np-hardness on bipartite graphs, we only need to get rid of these odd cycles and construct equivalent bipartite instances for each instance created in \Cref{thm:c_bd_bp}. 
	For this purpose, we introduce for each vertex $v^i \in V^i$, $i\in[k]$, an extra vertex $(v^i)'$. 
	We include the set of all these vertices in $V_\ell$ and assign weight $1$ for the leader and $M$ for the follower.
	Additionally, for each edge $\{u^i,v^i\}$, for all $i\in[k]$ we add a vertex $(u^iv^i)'$.
	We include the set of all these vertices in $V_f$ and assign weight $0$ for the leader and $M$ for the follower. 		
	Next, we replace each edge $\{u^i,v^i\}$ from the previous construction by edges $\set{u^i, (u^i)'}, \set{v^i, (v^i)'}, \set{(u^i)', (u^iv^i)'}$ and $\set{(v^i)', (u^iv^i)'}$ (see \Cref{fig:bo:bipartite:c_bd_bp}).
	
	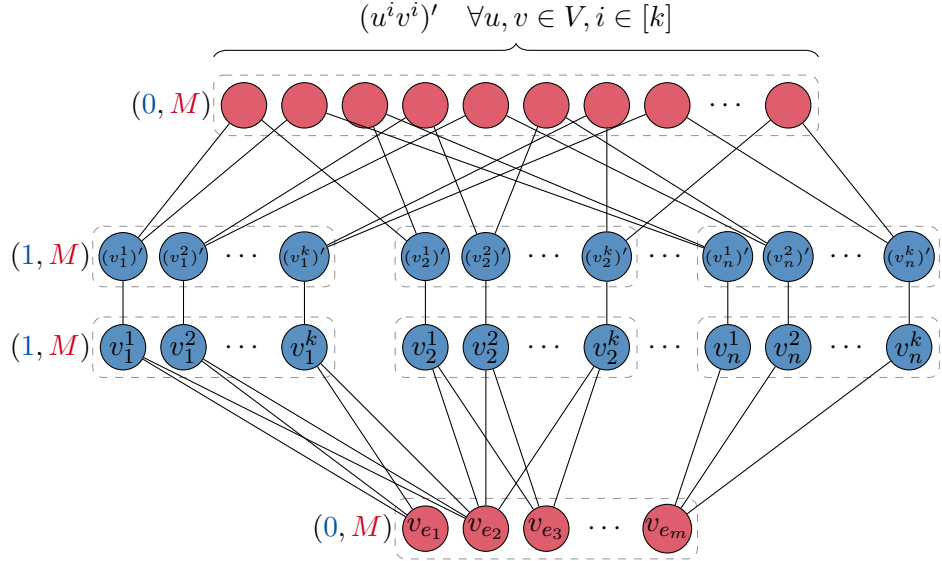
\begin{figure}[t]
		\begin{center}
			\begin{tikzpicture}[scale=0.8]
				\def\height{3}
				\def\heightprime{4.5}
				\def\edgeheight{7}
				
				\tikzset{
					leader_node/.style={circle, draw=black, fill=leader_color, minimum size=6mm, inner sep=0pt},
					follower_node/.style={circle, draw=black, fill=follower_color, minimum size=6mm, inner sep=0pt}
				}
					
				\draw[black!40, dashed, rounded corners] (-0.5,\height+0.5) -- (3.5,\height+0.5) -- (3.5,\height-0.5) -- (-0.5,\height-0.5) -- cycle;
				\draw[black!40, dashed, rounded corners] (4.5,\height+0.5) -- (8.5,\height+0.5) -- (8.5,\height-0.5) -- (4.5,\height-0.5) -- cycle;
				\draw[black!40, dashed, rounded corners] (9.5,\height+0.5) -- (13.5,\height+0.5) -- (13.5,\height-0.5) -- (9.5,\height-0.5) -- cycle;
				\draw[black!40, dashed, rounded corners] (4.5,0.5) -- (9.5,0.5) -- (9.5,-0.5) -- (4.5,-0.5) -- cycle;
				\draw[black!40, dashed, rounded corners] (1.5,\edgeheight+0.5) -- (11.5,\edgeheight+0.5) -- (11.5,\edgeheight-0.5) -- (1.5,\edgeheight-0.5) -- cycle;
				\draw[black!40, dashed, rounded corners] (-0.5,\heightprime+0.5) -- (3.5,\heightprime+0.5) -- (3.5,\heightprime-0.5) -- (-0.5,\heightprime-0.5) -- cycle;
				\draw[black!40, dashed, rounded corners] (4.5,\heightprime+0.5) -- (8.5,\heightprime+0.5) -- (8.5,\heightprime-0.5) -- (4.5,\heightprime-0.5) -- cycle;
				\draw[black!40, dashed, rounded corners] (9.5,\heightprime+0.5) -- (13.5,\heightprime+0.5) -- (13.5,\heightprime-0.5) -- (9.5,\heightprime-0.5) -- cycle;
				
				\node[leader_node] (v11) at (0,\height) {$v_1^1$};
				\node[leader_node] (v12) at (1,\height) {$v_1^2$};
				\node () at (2,\height) {$\cdots$};
				\node[leader_node] (v1k) at (3,\height) {$v_1^k$};
				
				\node[leader_node] (v21) at (5,\height) {$v_2^1$};
				\node[leader_node] (v22) at (6,\height) {$v_2^2$};
				\node () at (7,\height) {$\cdots$};
				\node[leader_node] (v2k) at (8,\height) {$v_2^k$};
				
				\node () at (9,\height) {$\cdots$};
				
				\node[leader_node] (v31) at (10,\height) {$v_n^1$};
				\node[leader_node] (v32) at (11,\height) {$v_n^2$};
				\node () at (12,\height) {$\cdots$};
				\node[leader_node] (v3k) at (13,\height) {$v_n^k$};
				\node[leader_node] (v11') at (0,\heightprime) {\tiny $(v_1^1)'$};
				\node[leader_node] (v12') at (1,\heightprime) {\tiny $(v_1^2)'$};
				\node () at (2,\heightprime) {$\cdots$};
				\node[leader_node] (v1k') at (3,\heightprime) {\tiny $(v_1^k)'$};
				
				\node[leader_node] (v21') at (5,\heightprime) {\tiny $(v_2^1)'$};
				\node[leader_node] (v22') at (6,\heightprime) {\tiny $(v_2^2)'$};
				\node () at (7,\heightprime) {$\cdots$};
				\node[leader_node] (v2k') at (8,\heightprime) {\tiny $(v_2^k)'$};
				
				\node () at (9,\heightprime) {$\cdots$};
				
				\node[leader_node] (v31') at (10,\heightprime) {\tiny $(v_n^1)'$};
				\node[leader_node] (v32') at (11,\heightprime) {\tiny $(v_n^2)'$};
				\node () at (12,\heightprime) {$\cdots$};
				\node[leader_node] (v3k') at (13,\heightprime) {\tiny $(v_n^k)'$};
				
				\node[follower_node] (ve1) at (5,0) {$v_{e_1}$};
				\node[follower_node] (ve2) at (6,0) {$v_{e_2}$};
				\node[follower_node] (ve3) at (7,0) {$v_{e_3}$};
				\node () at (8,0) {$\cdots$};
				\node[follower_node] (vem) at (9,0) {$v_{e_m}$};
				\node[follower_node] (e1) at (2,\edgeheight) {};
				\node[follower_node] (e2) at (3,\edgeheight) {};
				\node[follower_node] (e3) at (4,\edgeheight) {};
				\node[follower_node] (e4) at (5,\edgeheight) {};
				\node[follower_node] (e5) at (6,\edgeheight) {};
				\node[follower_node] (e6) at (7,\edgeheight) {};
				\node[follower_node] (e7) at (8,\edgeheight) {};
				\node[follower_node] (e8) at (9,\edgeheight) {};
				\node[follower_node] (e9) at (11,\edgeheight) {};
				\node () at (10,\edgeheight) {$\cdots$};
				\draw[decorate,decoration={brace,amplitude=5pt}] (1.5, \edgeheight+0.8) -- node[above=5pt,align=center]{$(u^iv^i)' \quad \forall u,v\in V, i\in[k]$} (11.5, \edgeheight+0.8);
				
				\node at (-1.2, \heightprime) {\ew{1}{M}};
				\node at (-1.2, \height) {\ew{1}{M}};
				\node at (3.8, 0) {\ew{0}{M}};
				\node at (0.8, \edgeheight) {\ew{0}{M}};

				\draw (e1) --(v11');
				\draw (e1) --(v21');
				\draw (e2) --(v11');
				\draw (e2) --(v31');
				\draw (e3) --(v21');
				\draw (e3) --(v31');
				\draw (e4) --(v12');
				\draw (e4) --(v22');
				\draw (e5) --(v12');
				\draw (e5) --(v32');
				\draw (e6) --(v22');
				\draw (e6) --(v32');
				\draw (e7) --(v1k');
				\draw (e7) --(v2k');
				\draw (e8) --(v3k');
				\draw (e8) --(v1k');
				\draw (e9) --(v3k');
				\draw (e9) --(v2k');
				
				\draw (v11) --(v11');
				\draw (v12) --(v12');
				\draw (v1k) --(v1k');
				\draw (v21) --(v21');
				\draw (v22) --(v22');
				\draw (v2k) --(v2k');
				\draw (v31) --(v31');
				\draw (v32) --(v32');
				\draw (v3k) --(v3k');
				
				\draw (v11) --(ve1);
				\draw (v12) --(ve1);
				\draw (v1k) --(ve1);
				\draw (v11) --(ve2);
				\draw (v12) --(ve2);
				\draw (v1k) --(ve2);
				\draw (v21) --(ve2);
				\draw (v22) --(ve2);
				\draw (v2k) --(ve2);			
				\draw (v21) --(ve3);
				\draw (v22) --(ve3);
				\draw (v2k) --(ve3);			
				\draw (v31) --(vem);
				\draw (v32) --(vem);
				\draw (v3k) --(vem);
			\end{tikzpicture}
		\end{center}
		\caption{\np-hardness reduction for the \bis on bipartite graphs when considering $(\costleader_b, \costfollower_b, \pessimistic)$ or $(\costleader_b, \costfollower_b, \optimistic)$. The blue vertices form the set $V_\ell$ whereas the red vertices form the set $V_f$.}
		\label{fig:bo:bipartite:c_bd_bp}
	\end{figure}
	
	Note that the resultant graph is bipartite with two parts being $\{v_e\mid e\in E(G)\}\cup \{(v^i)'\mid v\in V(G), i\in[k]\}$ and $\{v^i\mid v\in V(G),i\in [k]\} \cup \{(u^iv^i)'\mid u,v\in V(G), i\in[k]\}$.
	
	We further prove that the leader achieves objective value at least $1$ if and only if the original instance of the \textsc{Vertex Cover} problem has a vertex cover of size at most $k$.
	Note that the leader achieves objective value $0$ if and only if the follower's reaction satisfies $F\neq \emptyset$.
	In order to make sure that $F = \emptyset$, we claim that the leader has to choose at most one vertex from each $V^i$. 
	Consider the contrary.
	If two vertices $u^i$ and $v^i$ in $V^i$ belong to the leader's (optimal) action, then the leader cannot pick either of $(u^i)'$ and $(v^i)'$.
	Then a pessimistic follower would certainly pick $(u^iv^i)'$ in his solution set~$F$.
	Moreover, an optimistic or pessimistic follower with sum type objective function would also pick the available vertex $(u^iv^i)'$ into his solution, as this will increase his objective value by $M$.
	Thus, this construction works for the variants  $(\costleader_b, \costfollower_b, \pessimistic)$ and $(\costleader_b, \costfollower_s, \optimistic/\pessimistic)$ on bipartite graphs. 
	
	Let $S$ be a vertex cover of graph $G$ of size at most $k$.
	Then $L$ which includes, for each $i\in[k]$, a copy $v^i$ of a distinct vertex $v\in S$ and all other $(v^i)'$ copies whose neighbors have not been chosen earlier, is an optimal leader's action resulting in objective value $1$ for the leader.
	
	Similarly, if the constructed instance of \bis is a yes-instance, then there exists a leader's action $L$, which results in her objective value being $1$; here the cardinality $|L \cap \{v^i\mid v\in V(G), i\in[k]\}| \leq k$.
	Construct a set of the corresponding vertices in $G$.
	This set is a vertex cover of $G$ of size at most $k$, since for each vertex $v_e$, there is at least one neighbor in $L \cap \{v^i\mid v\in V(G), i\in[k]\}$. 
	
	Thus, the above construction proves \np-hardness on bipartite graphs for the cases $(\costleader_b, \costfollower_b, \pessimistic)$ and $(\costleader_b, \costfollower_s, \optimistic/\pessimistic)$.
	
	Moreover, the variant $(\costleader_b, \costfollower_b, \pessimistic)$ belongs to the class \np, as argued in \Cref{thm:c_bd_bp}. 
	
	To show the \np containment of the variant $(\costleader_b, \costfollower_s, \optimistic)$, we claim that $L\cup F\subseteq V_\ell\cup V_f$ serves as a valid certificate.
	First calculate, for the leader's action $L$, a maximum weighted independent set~$I$---with respect to weight function~$w_f$---in $G[S]$ where $S=\{v\in V_f \mid v\notin N(L)\}$; it is possible to do this in polynomial time as we are dealing with bipartite graphs.
	A follower's possible reaction $F$ to $L$ must satisfy $\sum_{v\in I} w_f(v) = \sum_{v\in F} w_f(v)$.
	If this holds, further check that $L\cup F$ is an independent set and calculate the leader's objective value $\argmin_{v\in L\cup F} w_\ell(v)$.
	Thus, it is possible to verify a yes-instance of the decision version of $(\costleader_b, \costfollower_s, \optimistic)$ in polynomial time. 
	Hence, the $(\costleader_b, \costfollower_s, \optimistic)$ variant of \bis is \np-complete on bipartite graphs.
	
	For the last variant $(\costleader_b, \costfollower_s, \pessimistic)$, the follower behaves in a pessimistic way and has sum type objective function.
	So his goal is to select a maximum weighted---with respect to function $w_f$---independent set~$F$ from $G[S]$ that additionally minimizes the value $\min_{v\in F} w_\ell(v)$.
	This can be achieved by doing the following:
	First, consider the vertices in $S$ in an ascending order of their $w_\ell$ values.
	Then one-by-one, check for each vertex $v\in S$, the maximum weighted independent set---with respect to function $w_f$---in $G[S]$ that includes the vertex $v$. 
	If the follower's total weight is the same as that of set~$I$ calculated above, then this serves as the optimal reaction of the follower towards $L$.
	If not, then calculate for the next vertex $v$ in $S$.
	This way, we can calculate the follower's optimal reaction to $L$ in polynomial time and check for the decision bounds of the problem. 
	Thus, the $(\costleader_b, \costfollower_s, \pessimistic)$ variant of \bis is \np-complete on bipartite graphs.
\end{proof}

Recall from \Cref{thm:general:P}, \bis for the variant $(\costleader_b, \costfollower_b, \optimistic)$ is polynomial-time solvable on general graphs.
Consequently, it is also polynomial-time solvable on bipartite graphs.

	\section{Conclusion and Future Directions} 
	\label{sec:conclusion}
	
		In this study, we considered partitioned-items bilevel optimization versions of the Independent Set and the Interval Selection problems. 
		We considered eight different variants emerging from selecting either sum type $(\costleader_s, \costfollower_s)$ or bottleneck type $(\costleader_b, \costfollower_b)$ objective functions for the leader and the follower, and further considering the optimistic $(\optimistic)$ or pessimistic $(\pessimistic)$ settings, which reflect the behavior of the follower towards the leader when faced with multiple choices for an optimal reaction.
		
		We studied the \BIS problem on the class of simple undirected graphs and also on bipartite graphs (see \Cref{results:overview}).
		We completely settled the complexity status of all eight variants of this problem on both of these graph classes.
		As anticipated, when going from simple undirected graphs to bipartite graphs, the computational complexity of most of the variants reduced by exactly one level of the polynomial hierarchy.
		However, interestingly, the variant $(\costleader_b, \costfollower_b,\pessimistic)$ disregarded this trend, and for this, the computational complexity of the problem remains the same. 
		
		We also studied the \bisch problem.
		Here we only considered the variants $(\costleader_s, \costfollower_s,\optimistic)$ and $(\costleader_s, \costfollower_s,\pessimistic)$, and we gave an algorithm that runs in $\mathcal{O}(n^4\log n)$ time.
		It would be interesting to know what complexity results hold for the other variants. 
		
		Some of the tractability results that we have presented involves exhaustively searching through all possible leader's actions and the follower's reactions to them.
		It would be nice to design algorithms that exploit some structure of the problem.
		
		For the future outlook, one could deviate from the partitioned-items bilevel problems and consider the settings wherein the sets of variables controlled by the leader and the follower are not disjoint.
	
	\subsection*{Acknowledgment} 
	The author would like to thank Dennis Fischer for some valuable discussions on this topic.
		
	\bibliographystyle{plain}
	\bibliography{bibliography}	
\end{document}